\documentclass[12pt,onecolumn,letter]{IEEEtran}
\usepackage{verbatim}
\usepackage{amsmath}
\usepackage{amsthm}
\usepackage{amssymb}
\usepackage{color}
\usepackage{bm}
\usepackage{multirow}

\usepackage[dvips]{graphicx}

\newtheorem{proposition}{Proposition} 
\newtheorem{theorem}{Theorem}
\newtheorem{corollary}{Corollary}
\newtheorem{remark}{Remark}
\newtheorem{lemma}{Lemma}

\newtheorem{definition}{Definition}

\def\Ker{\mathop{\rm Ker}}

\def\QED{\mbox{\rule[0pt]{1.5ex}{1.5ex}}}

\def\endproof{\hspace*{\fill}~\QED\par\endtrivlist\unskip}

\renewcommand{\qed}{\hfill \QED}

 \newenvironment{proofof}[1]{\vspace*{5mm} \par \noindent
         \quad{\it Proof of #1:\hspace{2mm}}}{\qed
}

\def\FF{\mathbb{F}}
\def\ZZ{\mathbb{Z}}
\def\rank{\mathop{\rm rank}}
\def\im{\mathop{\rm Im}}
\def\rank{\mathop{\rm rank}}

\def\Label#1{\label{#1}\ [\ \text{#1}\ ]\ }
\def\Label{\label}

\makeatletter
\def\mojiparline#1{
    \newcounter{mpl}
    \setcounter{mpl}{#1}
    \@tempdima=\linewidth
    \advance\@tempdima by-\value{mpl}zw
    \addtocounter{mpl}{-1}
    \divide\@tempdima by \value{mpl}
    \advance\kanjiskip by\@tempdima
    \advance\parindent by\@tempdima
}
\makeatother

\begin{document}

\title{Secrecy and Robustness for Active Attacks\\ in Secure Network Coding\\
and its Application to Network Quantum Key Distribution}

\author{Masahito Hayashi, 
Masaki Owari, Go Kato, and Ning Cai 
\thanks{The material in this paper was presented in part at the 2017 IEEE International Symposium on Information Theory (ISIT 2017),   Aachen (Germany), 25-30 June 2017.}
\thanks{Masahito Hayashi is with the Graduate School of Mathematics, Nagoya University, Japan. He is also with 
Shenzhen Institute for Quantum Science and Engineering, Southern University of Science and Technology
and the Centre for Quantum Technologies, National University of Singapore, Singapore
(e-mail:masahito@math.nagoya-u.ac.jp).
Masaki Owari is with Department of Computer Science, Faculty of Informatics, Shizuoka University, Japan 
(e-mail:masakiowari@inf.shizuoka.ac.jp).
Go Kato is with NTT Communication Science Laboratories, NTT Corporation, Japan
(e-mail:kato.go@lab.ntt.co.jp).
Ning Cai is with the School of Information Science and Technology, ShanghaiTech University
(e-mail: cai@gmx.de).} }

\markboth{M. Hayashi, M. Owari, G. Kato, and N. Cai: Secrecy and Robustness for Active Attack in Secure Network Coding}{}

\maketitle

\begin{abstract}
In network coding,
we discuss the effect of sequential error injection on information leakage.
We show that there is no improvement 
when the operations in the network are linear operations.
However, 
when the operations in the network contains non-linear operations,
we find a counterexample to improve Eve's obtained information.
Furthermore, we discuss the asymptotic rate in a linear network
under the secrecy and robustness conditions as well as 
under the secrecy condition alone.
Finally, we apply our results to network quantum key distribution,
which clarifies the type of network that enables us to realize secure long distance communication via
short distance quantum key distribution.
\end{abstract}


\section{Introduction}

Secure network coding offers a method for securely transmitting information from 
an authorized sender to an authorized receiver.
Cai and Yeung \cite{Cai2002} discussed the secrecy 
when the malicious adversary, Eve, wiretaps  a subset $E_E$ of all the channels in a network.
Using the universal hashing lemma \cite{bennett95privacy,HILL,hayashi11}, 
the papers \cite{Matsumoto2011,Matsumoto2011a} showed the existence of a secrecy code that works universally for any type of eavesdropper 
when the cardinality of $E_E$ is bounded.
In addition, the paper \cite{KMU} discussed the construction of such a code.
As another type of attack on information transmission via a network,
a malicious adversary contaminates the communication by 
changing the information on a subset $E_A$ of all the channels in the network.
Using an error correction, the papers \cite{Cai06a,Cai06,HLKMEK,JLHE} proposed a method to protect the message from contamination.
That is, we require that the authorized receiver correctly recovers the message, which is called robustness.
Now, for simplicity, we consider the unicast setting.
When the transmission rate from the authorized sender, Alice, to the authorized receiver, Bob, is $m_{0}$
and the rate of noise injected by Eve is $m_{1}$,
using the results published in \cite{JLKHKM,Jaggi2008}
the study \cite{JL} showed that
there exists a sequence of asymptotically correctable codes with the rate
$m_{0}-m_{1}$ if the rate of information leakage to Eve is less than $m_{0}-m_{1}$.

However, there is a possibility that the malicious adversary combines eavesdropping and contamination.
That is, contaminating a part of the channels, the malicious adversary might improve the ability of eavesdropping
while a parallel network offers no such a possibility \cite{ZKBJS1,ZKBJS2,KSZBJ}.
In fact, in arbitrarily varying channel model,
noise injection is allowed after Eve's eavesdropping, but  
Eve does not eavesdrop the channel after Eve's noise injection \cite{CN88,TJBK}\cite[Table I]{KJL}.
The studies \cite{KMU,Zhang} discussed the secrecy
when Eve eavesdrops the information transmitted on the channels in $E_E$ after noises are injected in $E_A$,
but they assumes that Eve do not know the information of the injected noise. 

In contrast, this paper discusses the secrecy when Eve 
adds artificial information to the information transmitted on the channels in $E_A$, 
eavesdrops the information transmitted on the channels in $E_E$, and 
estimates the original message from the eavesdropped information and the information of the injected noises.
We call this type of attack an active attack 
and call an attack without contamination a passive attack.
Specially, we call each of Eve's active operations a strategy.
Indeed, while the paper \cite{Yao2014} discusses robustness for an active attack,
it discusses secrecy only for a passive attack. 
When $E_A=E_E$ and any active attack is available for Eve, she is allowed to arbitrarily modify the information on the channels in $E_A$ sequentially based on the obtained information.

\begin{figure}[htbp]
\begin{center}
\includegraphics[scale=0.5]{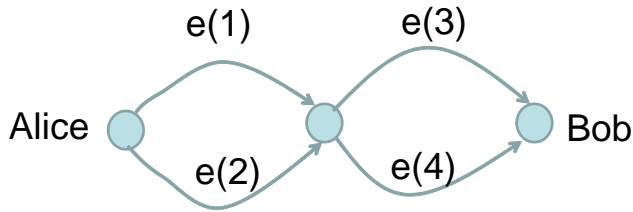}
\end{center}
\caption{One hop relay network.}
\Label{F1}
\end{figure}%

The aim of this paper is as follows.
Firstly, we show that no strategy can improve Eve's information
when any operation in the network is linear.
Then, we discuss a code that satisfies the need for  secrecy and robustness
when the transmission rate from Alice to Bob is $m_{0}$,
the rate of noise injected by Eve is $m_{1}$,
and the rate of information leakage to Eve is $m_{2}$.
In the asymptotic setting, we show the existence of such a secure protocol with rate $m_{0}-m_{1}-m_{2}$.
We discuss the asymptotic performance when 
only secrecy is considered.
When Alice and Bob share a small number of initial secret keys 
and can communicate with each other via a public channel, 
we do not impose robustness, but need the correctness only for passive attack.
In such a case, we show the existence of a secure protocol with the rate $m_{0}-m_{2}$.
This setting is useful for the secure communication over a network of quantum key distribution systems.

Quantum key distribution enables secure communication between two parties \cite{BB84}.
Recently, finite-length security analysis has been developed \cite{H-QKD,TWGR,HT1} even with multiple photons \cite{H-QKE2,H-N}.
Its commercial use has been well developed for limited transmission distance \cite{IDQ}.
However, it is very difficult to directly connect two distinct parties over long distances via quantum key distribution.
To realize long distance communication with quantum key distribution over short distances, 
this paper considers a method to connect them via a network composed of quantum key distribution over short distances. 
That is, we generate many pairs of shared secure keys on intermediate nodes by quantum key distribution, where each pair is composed of two nodes close to each other.
The secure keys shared by two nodes realize a secure channel between the two nodes.
Applying the method of network coding to these secure channels,
we realize secure communication between two distinct parties across a long distance.
We call this method network quantum key distribution.
When all the nodes are trusted, 
quantum key distribution guarantees the security of all the channels between the nodes.
However, there is a possibility that a part of nodes are occupied by Eve.
In this case, we do not impose robustness, but need the correctness only for passive attack.
Hence, for the security analysis of this case, we apply the above result to the network composed of these channels.

Next, to clarify the necessity of the linearity assumption, 
we discuss the ultimate performance on the one hop relay network (Fig. \ref{F1}) with the single shot case.
For this aim, 
we provide a counterexample of the non-linear coding of the binary case
on the network, in which, there exists a strategy to improve Eve's information.
A concrete description of this non-linear code on the one hop relay network (Fig. \ref{F1}) will be given in the main body of this paper.
In this example, even if Eve eavesdrops one edge before and after the intermediate node
as a passive attack,
she cannot recover the original message.
However, when she makes an active attack before the intermediate node,
she can recover the original message.
This example shows the importance of the assumption of linearity.
Similar unexpected properties for a nonlinear network error correcting code were reported in \cite{YYZ}.
Then, we show that the network code is limited to this counterexample
when we impose several natural secrecy conditions to the code on the one hop relay network (Fig. \ref{F1}) in the binary case.
This discussion shows that no code can guarantee the security over this type of active attack 
on the one hop relay network (Fig. \ref{F1}) in the binary case.
However, in the ternary case, there exists a code such that Eve cannot completely recover the message even with this type of active attack on the one hop relay network (Fig. \ref{F1}).
To discuss this problem, we introduce a new concept an ``anti-Latin square'', which is an opposite concept to a Latin square.

The remaining part of this paper is organized as follows.
Section \ref{S2} formulates our problem and shows the 
impossibility of Eve's eavesdropping under a linear network.
Section \ref{S3} discusses the asymptotic setting,
and show the achievability of the asymptotic rate $m_{0}-m_{1}-m_{2}$.
Section \ref{S5} discusses the asymptotic setting with secrecy without robustness.
Using the result of Section \ref{S5},
Section \ref{S6} considers the application of obtained results to 
network quantum key distribution
Section \ref{S10} applies the obtained result to multiple multicast network.
Section \ref{S4} discusses the ultimate performance on the one hop relay network.
In Section \ref{SCon}, we state the conclusion.


\section{Formulation of linear network and reduction theorem}\Label{S2}
\subsection{Single shot setting}
\subsubsection{Generic model}
In this subsubsection, we give a generic model, and discuss its relation with concrete network model in the latter subsubsections.
We consider the unicast setting of network coding on a network.
Assume that the authorized sender, Alice, intends to send information to the authorized receiver, Bob,
via the network.
We also assume that
Alice inputs the input variable $X$ in $\FF_q^{m_{{3}}}$ and Bob receives the output variable $Y_B$ in $\FF_q^{m_{{4}}}$.
We also assume that the malicious adversary, Eve, wiretaps the information $Y_E$ in $\FF_q^{m_{6}}$. 
Then, we adopt the model with matrices $K_B \in \FF_q^{m_{{4}}\times m_{{3}} }$ and $K_E \in \FF_q^{m_{6}\times m_{{3}} }$, in which,the variables $X$, $Y_B$, and $Y_E$ satisfy their relations
\begin{align}
Y_B=K_B X, \quad
Y_E=K_E X.\Label{F2}
\end{align}
We call this attack a {\it passive attack}.

In this paper, we address a stronger attack, in which, Eve injects noise $Z \in \FF_q^{m_{5}}$.
Hence, using matrices $H_B \in \FF_q^{m_{{4}}\times m_{5} }$ and $H_E \in \FF_q^{m_{6}\times m_{5}}$, 
we rewrite the relations \eqref{F2} as
  \begin{align}
\begin{aligned}
Y_B&=K_B X+ H_B Z, 
Y_E&=K_E X+ H_E Z, 
\end{aligned}
\Label{E3}
\end{align}
which is called a {\it wiretap and addition attack model}.
The $i$-th injected noise $Z_i$ (the $i$-th component of $Z$) is decided by a function $\alpha_i$ of $Y_E$.
In this paper, when a vector has the $j$-th component $x_j$,
the vector is written as $[x_j]_{1\le j \le a}$, where
the subscript $1\le j \le a$ expresses the range of the index $j$.
Thus, the set $\alpha=[\alpha_i]_{1\le i\le m_6}$ of the functions can be regarded as Eve's strategy,
and we call this attack an {\it active attack} with a strategy $\alpha$.
That is, a pair of a strategy and wiretap and addition attack model is called an active attack.
Here, we treat $K_B,K_E,H_B$, and $H_E$ as deterministic values, and denote the pairs $(K_B,K_E)$ and 
$(H_B,H_E)$ by $\bm{K}$ and $\bm{H}$, respectively.
Hence, our model is written as the triplet $(\bm{K}, \bm{H}, \alpha )$.
 We can consider several types for conditions for our model. 
As shown in the latter subsubsections, the triplet $(\bm{K}, \bm{H}, \alpha )$ is decided from the network topology and dynamics of the network.

We set the parameter $m_0$ as 
\begin{align}
\rank K_B= m_{0} ,
\Label{1-6}
\end{align}
and assume the ranks of $H_B$ and $K_E$ as
\begin{align}
\rank H_B= m_{1}, ~ \rank K_E= m_{2} .
\Label{1-6X}
\end{align}
Then, the parameters are summarized in Table \ref{hikaku}.

\begin{table}[htpb]
  \caption{Channel parameters}
\Label{hikaku}
\begin{center}
  \begin{tabular}{|c|l|} 
\hline
\multirow{2}{*}{$m_0$} & Rank of the channel from Alice\\
&  to Bob, i.e., $\rank K_B$ \\
\hline
$m_1$ & Rank of Eve's injected information ($\rank H_B$)\\
\hline
$m_2$ & Rank of Eve's wiretapped information ($\rank K_E$)\\
\hline
$m_3$ & Dimension of Alice's input information\\
\hline
$m_4$ & Dimension of Bob's observed information \\
\hline
$m_5$ & Dimension of Eve's injected information \\
\hline
$m_6$ & Dimension of Eve's wiretapped information \\
\hline
  \end{tabular}
\end{center}
\end{table}

We can consider two types for conditions for our model. 

\begin{definition}
For any value of $K_E x$,
there uniquely exists $y \in \FF_q^{m_6} $ such that
\begin{align}
y= K_E x+ H_E \alpha(y).\Label{Uni}
\end{align}
This condition is called the {\it uniqueness condition}.
\end{definition}

Although $\alpha_i $ is a function of the vector $[Y_{E,j}]_{1\le j \le m_6}$,
it is natural that $\alpha_i $ is decided by a part of Eve's observed variables
when we take the causality with respect to $\alpha$ into account.
Since the decision of the injected noise does not depend on the results of the decision,
we introduce the causal condition.

Hence, we choose the subset $w_i \subset \{1, \ldots, m_6\}$ such that
the function $\alpha_i $ is given as a function of the vector $[Y_{E,j}]_{j \in w_i}$.

\begin{definition}\Label{Def3}
Generally, the function $\alpha_i $ is given as a function of a part of component of the vector $[Y_{E,j}]_{1, \le j \le  m_6}$.
To clarify this point, 
we choose the subset $w_i \subset \{1, \ldots, m_6\}$ such that
the function $\alpha_i $ is given as a function of the vector $[Y_{E,j}]_{j \in w_i}$.
\begin{description}
\item[(A1)] The relation $ H_{E;j,i} = 0$ holds for $j \in w_i$.
(This condition means that
the $j$-th eavesdropping is done after the $i$-th injection for $j \in w_i$.)
\item[(A2)] 
The relation $ w_{1} \subset w_{2} \subset \ldots \subset 
w_{m_5}$ holds.
\end{description}
This condition is called the {\it causal condition}.
\end{definition}

Then, we have the following lemma.

\begin{lemma}\Label{LL2}
When the triplet $(\bm{K},\bm{H},\alpha)$ satisfies the causal condition, 
it satisfies the uniqueness condition.
\end{lemma}

\begin{proof}
When the causal condition holds,
we show the fact that $y_{j'}$ is given as a function of $K_E x$ for any $j' \in w_{i}$
by induction with respect to the index $i =1, \ldots, m_5$, which expresses the order of the injected information.
This fact yields the uniqueness condition.

For $j \in w_1$, we have$y_j= (K_E x)_j$ because $(H_E \alpha(y))_{j}$ is zero.
Hence, the statement with $j=1$ holds.
We choose $j \in w_{i+1}\setminus w_{i}$.
Let $z_{i'}$ be the $i'$-th injected information.
Due to Conditions (A1) and (A2),
$y_j- (K_E x)_j$ is a function of $z_{1}, \ldots, z_{i}$.
Since the causal condition guarantees that
$z_{1}, \ldots, z_{i}$ are functions of 
$[y_{j'}]_{j' \in w_{i}}$,
$z_{1}, \ldots, z_{i}$ are functions of $K_E x$.
Then, we find that $y_j$ is given as a function of 
$K_E x$ for any 
$j \in w_{i+1}\setminus w_{i}$.
That is, the triplet $(\bm{K},\bm{H},\alpha)$ satisfies the uniqueness condition.
\end{proof}

Now, we consider an equivalent condition to the uniqueness condition holds 
when $\alpha$ is given as a linear function, i.e., $\alpha(y)=Gy$ for a matrix $G$.
Equation \eqref{Uni} is equivalent to the equation 
$(I- H_E G)y= K_E x$.
Hence, the uniqueness is equivalent to the invertability of 
the matrix $I- H_E G$.
In fact, if the causal condition does not hold,
the matrix $I- H_E G$ is not invertible.

\subsubsection{Construction of passive attack model from directed graph}\Label{PDG}
Next, we discuss how we can obtain the generic passive attack model \eqref{F2} from 
a concrete network structure.
We consider the unicast setting of network coding on a network, which is given as 
a graph $({V}, {E})$ with direction, 
where the set ${V}$ of vertices expresses the set of nodes
and the set ${E}:=\{e(1), \ldots, e(m_7)\}$ of edges expresses the set of communication channels, where 
a communication channel means a packet in network engineering.
Here, the directed graph $({V}, {E})$ is not necessarily acyclic.
When a channel transmits information from a node $u\in {V}$ to another node $v\in {V}$, it is written as $(u,v) \in {E}$.

We assume that the transmission on the edge starts at the tail node of the edge in the order of the numbers assigned to the edges, which is called the {\it partial time-ordered condition}.
In the single-use transmission,
the source node has several elements of $\FF_q$
and sends each of them via its outgoing edges 
in the order of assigned number of edges,
where $\FF_q$ is a finite field whose order is a power $q$ of the prime $p$.
Each intermediate node keeps received information via incoming edges.
Then, for each outgoing edge,
the intermediate node calculates one element of $\FF_q$ 
from previously received information,
and sends it via the outgoing edge.
That is, every outgoing information from a node $v(i)$ via a channel $e(j)$ 
depends only on the incoming information into the node $v(i)$ via channels 
$e(j')$ such that $j'<j$.
The operations on all nodes are assumed to be linear on the finite field $\FF_q$ with prime power $q$.
Bob receives the information $Y_B$ in $\FF_q^{m_{4}}$ 
on the edges of a subset $E_B:=\{e(\zeta_B(1)), \ldots, e(\zeta_B(m_{4}))\}\subset E$,
where $\zeta_B$ is a strictly increasing function from $\{1, \ldots, m_4\}$ to $\{1, \ldots, m_7\}$.

Let $\tilde{X}_j$ be the information on the edge $e(j)$.
In the following, we describe the information on the $m_8:=m_7-m_3$ edges that are not directly linked to the source node.
When the edge $e(j)$ is a outgoing edge of the node $v(i)$,
the information $\tilde{X}_{j}$ is given as a linear combination of 
the information on the edges incoming to the node $v(i)$.
We have coefficients $\theta_{j,j'}$ such that
$\tilde{X}_{j}= \sum_{j'}\theta_{j,j'} \tilde{X}_{j'}$, where
$\theta_{j,j'} $ is zero 
unless $ e(j')$ is not an edge incoming to $v(i)$.
The partial time-ordered condition implies that 
\begin{align}
\theta_{j,j'}=0 \hbox{ for } j' \ge j. \Label{GTY}
\end{align}
Now, we define $m_8$ $m_7 \times m_7$ matrices.
That is, we define the $j$-th $m_7 \times m_7$ matrix $M_j$ as follows.
The $j+m_3$-th column vector of the matrix $M_j$ is defined by $[\theta_{j+m_3,j'}]_{1\le j'\le m_7}$.
The remaining part of $M_j$ is defined as the identity matrix.
Then, we have
\begin{align}
Y_{B,j}= \sum_{i=1}^{m_3}(M_{m_8}\cdots M_1)_{\zeta_B(j),i} X_i\Label{KOG}
\end{align}
While the output of the matrix $M_{m_8}\cdots M_1$ takes values in $\FF_q^{m_7}$,
we focus the projection $P_B$ to the subspace $\FF_q^{m_4}$ that corresponds to the $m_4$ components observed by Bob.
That is, $P_B$ is a $m_4 \times m_7$ matrix to satisfy $P_{B;i,j}=\delta_{\zeta_B(i),j}$.
Similarly, 
we use the projection $P_A$ (a $m_7 \times m_3$ matrix) as $P_{A;i,j}=\delta_{i,j}$.
Due to \eqref{KOG}, the matrix $K_B:=P_B  M_{m_8}\cdots M_1 P_A$ satisfies the first equation in \eqref{F2}.

The malicious adversary, Eve, wiretaps the information $Y_E$ in $\FF_q^{m_{6}}$ 
on the edges of a subset $E_E:=\{e(\zeta_E(1)), \ldots, e(\zeta_E(m_{6}))\}\subset E$,
where $\zeta$ is a strictly increasing function from 
$\{1, \ldots, m_6\}$ to $\{1, \ldots, m_7\}$.
Then, we have
\begin{align}
Y_{E,j}= \sum_{i=1}^{m_3}(M_{m_8}\cdots M_1)_{\zeta_E(j),i} X_i\Label{KOG2}
\end{align}
We employ the projection $P_E$ (a $m_6 \times m_7$ matrix) to the subspace $\FF_q^{m_6}$ that corresponds to the $m_6$ components eavesdropped by Eve.
That is, $P_{E;i,j}= \delta_{\zeta(i),j}$.
Then, we obtain the matrix $K_E$ as $P_E  M_{m_8}\cdots M_1 P_A$.
Due to \eqref{KOG}, the matrix $K_E:=P_E  M_{m_8}\cdots M_1 P_A$ satisfies the second equation in \eqref{F2}.

In summary 
the topology and dynamics (operations on the intermediate nodes) of the network, including the places of attached edges decides
the graph $(V,E)$, the coefficients $\theta_{i,j}$, and functions $\zeta_B,\zeta_E$,
uniquely gives the two matrices $K_B$ and $K_E$.
Here, we emphasize that we do not assume the acyclic condition for the graph $({V}, {E})$.
That is, due to the partial time-ordered condition,
we can uniquely define our matrices $K_B$ and $K_E$, which is a similar way to \cite[Section V-B]{ACLY}\footnote{%
$\Lambda$ of Ahlswede-Cai-Li-Yeung corresponds to
the number of edges that are not connected to the source node in our paper.}.

\subsubsection{Construction of active attack model from directed graph}
We construct the generic active attack model from 
a concrete network structure.
We assume that Eve injects the noise in a part of edges $E_A \subset E$
as well as eavesdrops the edges $E_E$ and 
assume the condition 
\begin{align}
E_A \cap \{1, \ldots, m_3\}= \emptyset.\Label{emp}
\end{align}
In fact, when the condition \eqref{emp} does not hold,
the following new graph satisfies the condition \eqref{emp}.
We add new vertexes on the edges $\{e(1), \ldots ,e(m_3)\}$,
and divides these edges into two edges.
Then, these tail parts keep the original numbers,
and these head parts are assigned the numbers $\{m_3+1,\ldots, 2 m_3 \}$.
The numbers of the remaining edges are changed by adding $m_3$.
This modified graph satisfies the condition \eqref{emp}.

The elements of the subset $E_A$ is expressed as
$E_A=\{e(\eta(1)), \ldots, e(\eta(m_{5}))\}$ 
by using a function $\eta$ from 
$\{1, \ldots, m_5\}$ to $\{1, \ldots, m_7\}$.
To give the matrices $H_B$ and $H_E$,
modifying the matrix $M_j$,
we define the new matrix $M_j'$ as follows
The $j+m_3$-th column vector of the new matrix $M_j'$ is defined by 
$[\theta_{j+m_3,j'}+\delta_{j+m_3,j'}]_{1\le j'\le m_7}$.
The remaining part of $M_j'$ is defined as the identity matrix.
Then, we have
\begin{align}
Y_{B,j}= &
\sum_{i=1}^{m_3}(M_{m_8}\cdots M_1)_{\zeta_B(j),i} X_i
+
\sum_{i'=1}^{m_5}(M_{m_8}'\cdots M_1')_{\zeta_B(j),\eta(i')} Z_{i'} 
\Label{KOG3}\\
Y_{E,j}= &
\sum_{i=1}^{m_3}(M_{m_8}\cdots M_1)_{\zeta_E(j),i} X_i
+
\sum_{i'=1}^{m_5}(M_{m_8}'\cdots M_1'- I)_{\zeta_E(j),\eta(i')} Z_{i'}.
\Label{KOG4}
\end{align}
When Eve eavesdrops the edges $E_E \cap E_A$,
she obtains the information on $E_E\cap E_A$ before her noise injection.
Hence, to express her obtained information on $E_E\cap E_A$,
we need to substract  her injected information on $E_E\cap E_A$.
Hence, we need $-I$ in the second term of \eqref{KOG4}.
We introduce the projection $P_{E,A}$ (an $m_7 \times m_5$ matrix) as $P_{A;i,j}=\delta_{i,\eta(j)}$.
Due to \eqref{KOG3} and \eqref{KOG4}, 
the matrices 
$H_B:=P_B  M_{m_8}'\cdots M_1' P_{E,A}$ 
and
$H_E:=P_E  (M_{m_8}'\cdots M_1'-I) P_{E,A}$ satisfy
conditions \eqref{E3} with the matrices
$K_B$ and $K_E$, respectively.
This model ($K_B$, $K_E$, $H_B$, $H_E$)
is called the {\it wiretap and addition model}
determined by $(V,E)$ and $(E_E, E_A,\{\theta_{i,j}\})$, which expresses 
the topology and dynamics.

To discuss the active attack, we consider the condition for the strategy $\alpha$
in addition to the wiretap and addition attack model.
One may assume that 
the tail of the edge $e(j)$ sends the information to the edge $e(j)$
after the head of the edge $e(j-1)$ receives the information to the edge $e(j-1)$,
which we call the {\it full time-ordered condition}.
However, the full time ordered condition does not hold in general even when we reorder the numbers assigned to the edges.
Although we can discuss the active attack only with the partial time-ordered condition,
we discuss it under the full time-ordered condition first. 

When the full time-ordered condition holds,
the function $\eta$ is a strictly increasing function from 
$\{1, \ldots, m_5\}$ to $\{1, \ldots, m_7\}$.
Since Eve can choose the information to be added on the edge $e(i) \in E_A$
based on the obtained information $Y_E$,
the added error $Z_i$ is given as a function $\alpha_i $ of 
the vector $[Y_{E,j}]_{j \in w_i}$ with $w_i:=\{j| \eta(i) \ge \zeta_E(j)\}$.
Since the function $\eta$ is strictly increasing, Condition (A2) for the causal condition holds.
Since the relation \eqref{GTY} implies that 
$M_{m_8}'\cdots M_1'- I$ is a lower triangular matrix with zero diagonal elements,
the strictly increasing property of $\eta$ yield that
\begin{align}
H_{E;j,i}=0 \hbox{ when } \eta(i) \ge \zeta(j) \Label{F14},
\end{align}
which implies Condition (A1) for the causal condition.
In this way,
the full time-ordered condition implies the causal condition.

However, in the realistic setting,
it is possible that Eve might intercept (i.e., wiretap and contaminate) the information of an edge
before the head node of the previous edge receives the information on the edge.
Hence, we consider the case when the partial time-ordered condition holds, but
the full time-ordered condition does not necessarily hold\footnote{%
For an example, we consider the following case.
Eve gets the information on the first edge. Then, she gets 
the information on the second edge before she hands over 
the information on the first edge to the tail node of the first edge.
In this case, she can change the information on the first edge
based on the information on the first and second edges. Then, the time-ordered condition \eqref{F14} does not hold.}.
That is, the function $\eta$ from $\{1, \ldots, m_5\}$ to $E $ is injective but is not necessarily monotone increasing.
Then, we choose the sets $w_i$ to satisfy Condition (A2) for the causal condition 
so that the added error $Z_i$ is given as a function $\alpha_i $ of the vector $[Y_{E,j}]_{j \in w_i}$.
The partial time-ordered condition implies the relation
\begin{align}
j < \gamma(i):=\min \{\zeta_E(j')|   \theta_{\zeta_E(j'),i}\neq 0\}
\Label{LGY}
\end{align}
for $j \in w_{i}$,
which implies the following condition; 
For $ j \in w_i$,
there is no sequence $j=j_1>j_2, \ldots >j_l=\eta(i) $ 
such that 
\begin{align}
\theta_{j_{i},j_{i+1}}\neq 0 .
\end{align}
This condition implies Condition (A1) for the causal condition.  
That is, even when the full time-ordered condition does not hold,
the causal condition can be naturally derived.

Now, we consider the optimal choice of $\eta, \{w_i\}$ for Eve.
That is, we choose the subset $w_i$ as large as possible under the partial time-ordered condition and Condition (A2).
Then, we choose the bijective function $\eta_{o}$ from $\{1, \ldots, m_5\}$
to the set of indexes of elements of $E_A$ 
such that $\gamma \circ \eta_o$ is monotone increasing.
Then, we define $w_{o,i}:=\{ j| \zeta_E(j) < \gamma(\eta_o(i))\}$,
which satisfies Conditions (A1) and (A2) for the causal condition.  
Further, 
for the above choice $ \eta, \{w_i\}$, the condition \eqref{LGY} implies 
$w_{\eta\circ \eta_o^{-1}(i)}\subset w_{o,i} $, i.e., 
$w_{o,i} $ is the largest subset under the partial time-ordered condition and Condition (A2),
which shows the optimality of 
$\eta_o, \{w_{o,i}\}$.
Although the choice of $\eta_o$ is not unique, 
the choice of $w_{o,\eta_o^{-1}(i)}$ for $e(i) \in E_A$
 is unique.

\subsubsection{Examples}
\Label{F1Ex}
In this subsubsection, as an example, we consider the network given in Figs. \ref{F3} and \ref{F1B}.
Alice sends the variables $X_1,\ldots, X_4 \in \FF_q$ 
to nodes $v(1),v(2),v(3),$ and $v(4)$
via the edges
$e(1),e(2),e(3)$, and $e(4)$, respectively.
The edges $e(5), \ldots, e(12)$
send the elements received in the tail node.
The edges $e(13)$ and $e(14)$ 
send the sum of two elements received in the tail node.
The received elements via the edges 
$e(11), e(12), e(13),$ and $e(14)$
are written as $Y_{B,1}, Y_{B,2}, Y_{B,3}$, and $Y_{B,4}$, respectively.
 Then, the matrix $K_B$ is given as 
\begin{align}
K_B= 
\left(
\begin{array}{cccc}
1 & 0 & 0 & 0 \\
0 & 0 & 1 & 0 \\
1 & 1 & 0 & 0 \\
0 & 0 & 1 & 1 
\end{array}
\right).
\end{align}
Then, $m_0=4$.

Now, we assume that Eve eavesdrops and contaminates
the edges 
$e(1), e(2), e(6), e(7)$, and $e(13)$.
We denote the observed information 
and the injected information
on the edges $e(1), e(2), e(6), e(7)$, and $e(13)$
by
$Y_{E,1}, Y_{E,2}, Y_{E,3}, Y_{E,4},Y_{E,5}$
and
$Z_{1}, Z_{2}, Z_{3}, Z_{4},Z_{5}$.
In Fig. \ref{F3A},
Eve adds $Z_{1}, Z_{2}, Z_{3}, Z_{4},Z_{5}$ in edges $e(1), e(2), e(6), e(7)$, and $e(13)$.
Here, Eve injects noises 
Then, the matrices $H_B$, $K_E$, and $H_E$ are given as 
\begin{align}
H_B= 
\left(
\begin{array}{ccccc}
1 & 0 & 0 & 0 & 0\\
0 & 0 & 0 & 0 & 0\\
1 & 1 & 1 & 1 & 1\\
0 & 0 & 0 &  0& 0 
\end{array}
\right), \quad
K_E= 
\left(
\begin{array}{cccc}
0 & 1 & 0 & 0 \\
1 & 0 & 0 & 0 \\
0 & 1 & 0 & 0 \\
1 & 0& 0 & 0 \\
1 & 1 & 0 & 0 
\end{array}
\right), \quad
H_E= 
\left(
\begin{array}{ccccc}
0 & 0 & 0 & 0 &0\\
0 & 0 & 0 & 0 &0\\
0 & 1 & 0 &  0 & 0 \\
1 & 0 & 0 & 0 &0\\
1 & 1 & 1 &  1 & 0 
\end{array}
\right).
\end{align}
Then, $\rank H_B=\rank K_E=2$.
Eve can choose the function $\eta$ as 
\begin{align}
\eta(1)=2,
\eta(2)=1,
\eta(3)=7,
\eta(4)=6,
\eta(5)=13,
\end{align}
and choose the subsets $w_i$ as
\begin{align}
w_1=w_2=\{ 1,2 \},
w_3=w_4=\{1,2,3,4\},
w_5=\{1,2,3,4,5\}.
\end{align}
This case satisfies Conditions (A1) and (A2).
Hence, this model satisfies the causal condition.
Lemma \ref{LL2} guarantees that it also satisfies the uniqueness condition.

\begin{figure}[htbp]
\begin{center}
\includegraphics[scale=0.4]{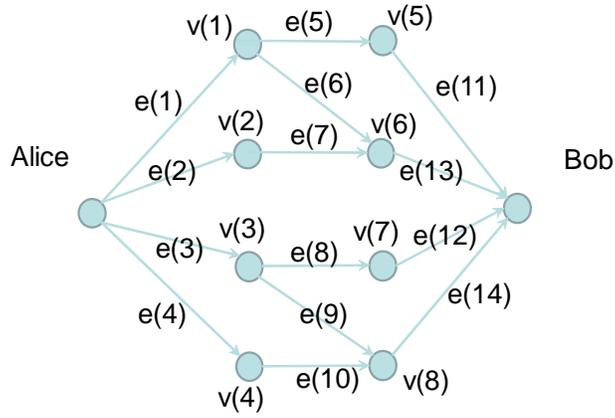}
\end{center}
\caption{Network of Subsubsection \ref{F1Ex} with name of edges}
\Label{F3}
\end{figure}%

\begin{figure}[htbp]
\begin{center}
\includegraphics[scale=0.4]{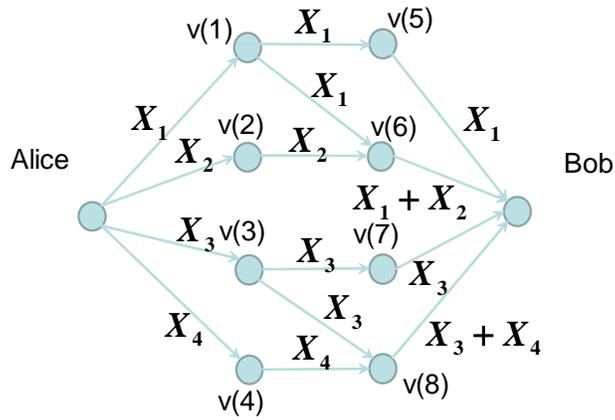}
\end{center}
\caption{Network of Subsubsection \ref{F1Ex} with network flow}
\Label{F1B}
\end{figure}%

\begin{figure}[htbp]
\begin{center}
\includegraphics[scale=0.4]{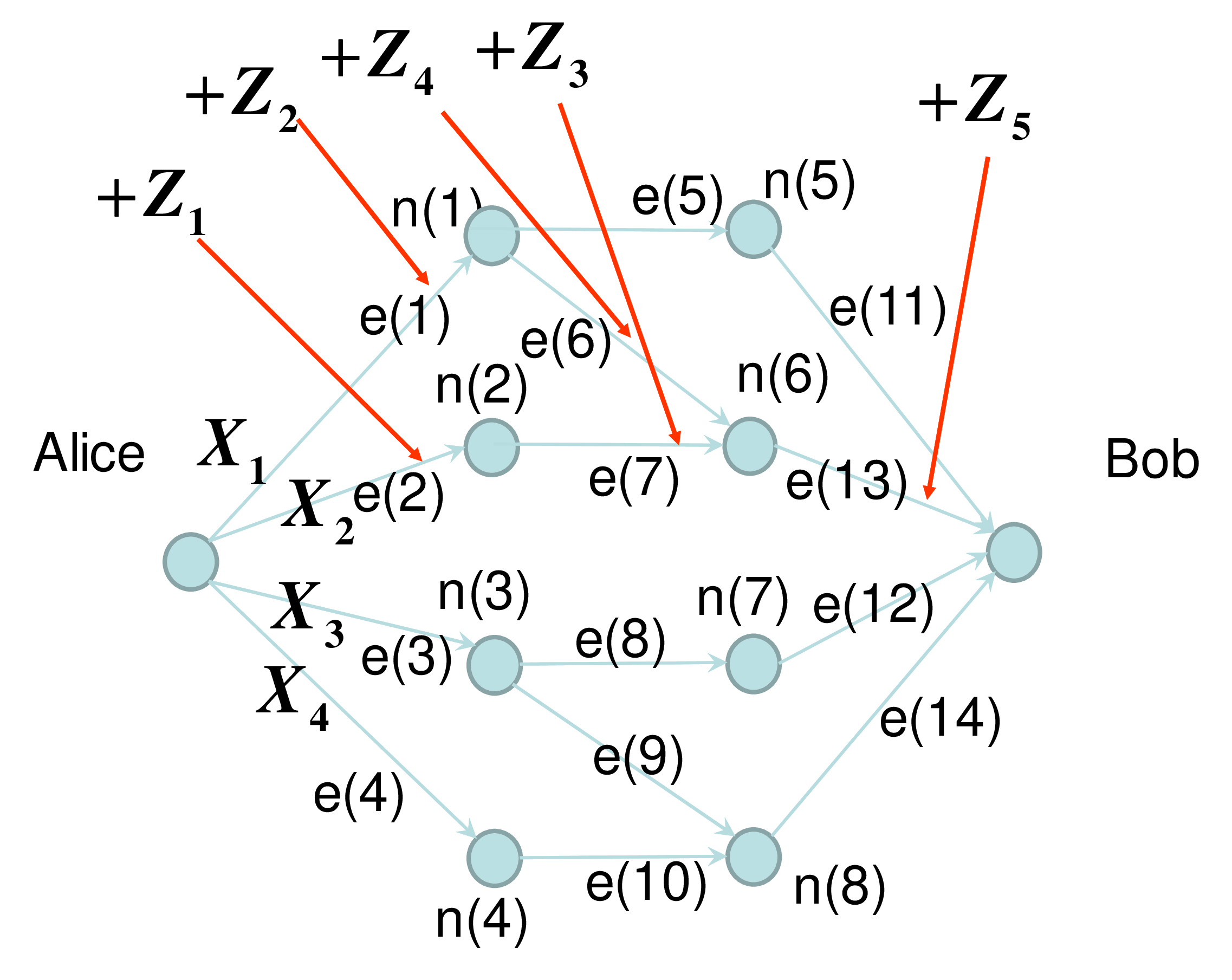}
\end{center}
\caption{Network of Subsubsection \ref{F1Ex} with addition attack}
\Label{F3A}
\end{figure}%


\subsubsection{Wiretap and replacement model}
In the above subsubsections, 
we have 
discussed the case when 
Eve injects the noise in the edges $E_A$ as well as eavesdrops the edges $E_E$.
In this subsubsection,
we assume that $E_A=E_E$
and
Eve eavesdrops the edges $E_E$ and replaces the information on the edges 
$E_A$ by other information.
While this assumption implies $m_5=m_6$ and
and the image of $\eta$ equals the image of $\zeta_E$, 
the function $\eta$ does not necessarily equal the function $\zeta_E$
because the order that Eve sends her replaced information to the heads of edges
does not necessarily equal the order that Eve intercepts the information on the edges.
This case also belongs to general wiretap and addition model \eqref{E3} as follows.
Modifying the matrix $M_j$,
we define the new matrix $M_j''$ as follows.
When there is an index $i$ such that $\zeta_E(i)=j$, 
the $j+m_3$-th column vector of the new matrix $M_j''$ is defined by 
$[\delta_{j+m_3,j'}]_{1\le j'\le m_7}$ and
the remaining part of $M_j''$ is defined as the identity matrix.
Otherwise, $M_j''$ is defined to be $M_j$.
Also, we define another matrix $F$ as follows.
The $\zeta_E (i)$-th column vector of the new matrix $F$ is defined by 
$[\theta_{\zeta_E(i),j'}]_{1\le j'\le m_7}$ and
the remaining part of $F$ is defined as the identity matrix.
Under the condition \eqref{emp},
we have
\begin{align}
Y_{B,j}= &
\sum_{i=1}^{m_3}(M_{m_8}''\cdots M_1'')_{\zeta_B(j),i} X_i
+
\sum_{i'=1}^{m_5}(M_{m_8}''\cdots M_1'')_{\zeta_B(j),\eta(i')} Z_{i'} 
\Label{KOG5}\\
Y_{E,j}= &
\sum_{i=1}^{m_3}(F M_{m_8}''\cdots M_1'')_{\zeta_E(j),i} X_i
+
\sum_{i'=1}^{m_5}(F M_{m_8}''\cdots M_1'')_{\zeta_E(j),\eta(i')} Z_{i'}.
\Label{KOG6}
\end{align}
Then, we choose 
matrices $K_B'$, $K_E'$, $H_B'$, and $H_E'$ as
$K_B':=P_B  M_{m_8}''\cdots M_1'' P_{A}$, 
$K_E':=P_E  FM_{m_8}''\cdots M_1'' P_{A}$, 
$H_B':=P_B  M_{m_8}''\cdots M_1'' P_{E}^T$, 
and
$H_E':=P_E F M_{m_8}''\cdots M_1'' P_{E}^T$,
which satisfy conditions \eqref{E3} due to \eqref{KOG5} and \eqref{KOG6}.
This model ($K_B'$, $K_E'$, $H_B'$, $H_E'$) is called the {\it wiretap and replacement model}
determined by $(V,E)$ and $(E_E, \{\theta_{i,j}\})$.

Next, we discuss the strategy $\alpha'$ under the matrices $K_B'$, $K_E'$, $H_B'$, and $H_E'$ such that
the added error $Z_i$ is given as a function $\alpha_i' $ of the vector $[Y_{E,j}]_{j \in w_i}$.
Since the decision of the injected noise does not depend on the results of the decision,
we impose the causal condition defined in Definition \ref{Def3} for the subsets $w_i$.

When the relation $j \in w_i$ holds with $\zeta_E(j)=\eta(i)$,
a strategy $\alpha'$ on 
the wiretap and replacement model 
($K_B'$, $K_E'$, $H_B'$, $H_E'$)
determined by $(V,E)$ and $(E_E, \{\theta_{i,j}\})$
is written by another strategy $\alpha$ on 
the wiretap and addition model
$K_B$, $K_E$, $H_B$, and $H_E$
determined by $(V,E)$ and $(E_E, E_E,\{\theta_{i,j}\})$,
which is defined as
$\alpha[Y_{E,j'}]_{j' \in w_i}:=
\alpha'[Y_{E,j'}]_{j' \in w_i}- Y_{E,j}$.
In particular, 
due to the condition \eqref{GTY},
the optimal choice $\eta_o,\{w_{o,i}\}$ 
under the partial time-ordered condition
satisfies 
the relation $j \in w_{o,i}$ holds with $\zeta_E(j)=\eta_o(i)$.
That is, under the partial time-ordered condition,
the strategy on the wiretap and replacement model 
can be written by another strategy on the wiretap and addition model.

However, if there is no synchronization among vertexes,
Eve can inject the replaced information to the head of an edge 
before the tail of the edge sends the information to the edge.
Then, the partial time-ordered condition does not hold.
In this case, the relation $j \in w_i$ does not necessarily hold with $\zeta_E(j)=\eta(i)$.
Hence, a strategy $\alpha'$ on 
the wiretap and replacement model 
($K_B'$, $K_E'$, $H_B'$, $H_E'$)
cannot be necessarily written as another strategy on 
the wiretap and addition model
($K_B$, $K_E$, $H_B$, $H_E$).

To see this fact, we discuss an example given in Subsubsection \ref{F1Ex}.
In this example, 
the network structure of the wiretap and replacement attack model 
is given by Fig. \ref{F3B}.
In Fig. \ref{F3B}, we change the order of replacement.
The function $\eta$ is given as 
\begin{align}
\eta(1)=6,
\eta(2)=7,
\eta(3)=13,
\eta(4)=1,
\eta(5)=2,
\end{align}
and choose the subsets $w_i$ as
\begin{align}
w_1=w_2=\{ 1,2 \},
w_3=w_4=w_4=\{1,2,5\}.
\end{align}
This case satisfies Conditions (A1) and (A2).
Hence, this attack satisfies the causal condition.
In contrast, this strategy cannot be written as a causal strategy on 
the wiretap and addition model
$K_B$, $K_E$, $H_B$, and $H_E$
because the replaced information on the edges $e(1)$ and $e(2)$
depend on the obtained information $Y_5$ on the edge $e(13)$.


\begin{figure}[htbp]
\begin{center}
\includegraphics[scale=0.38]{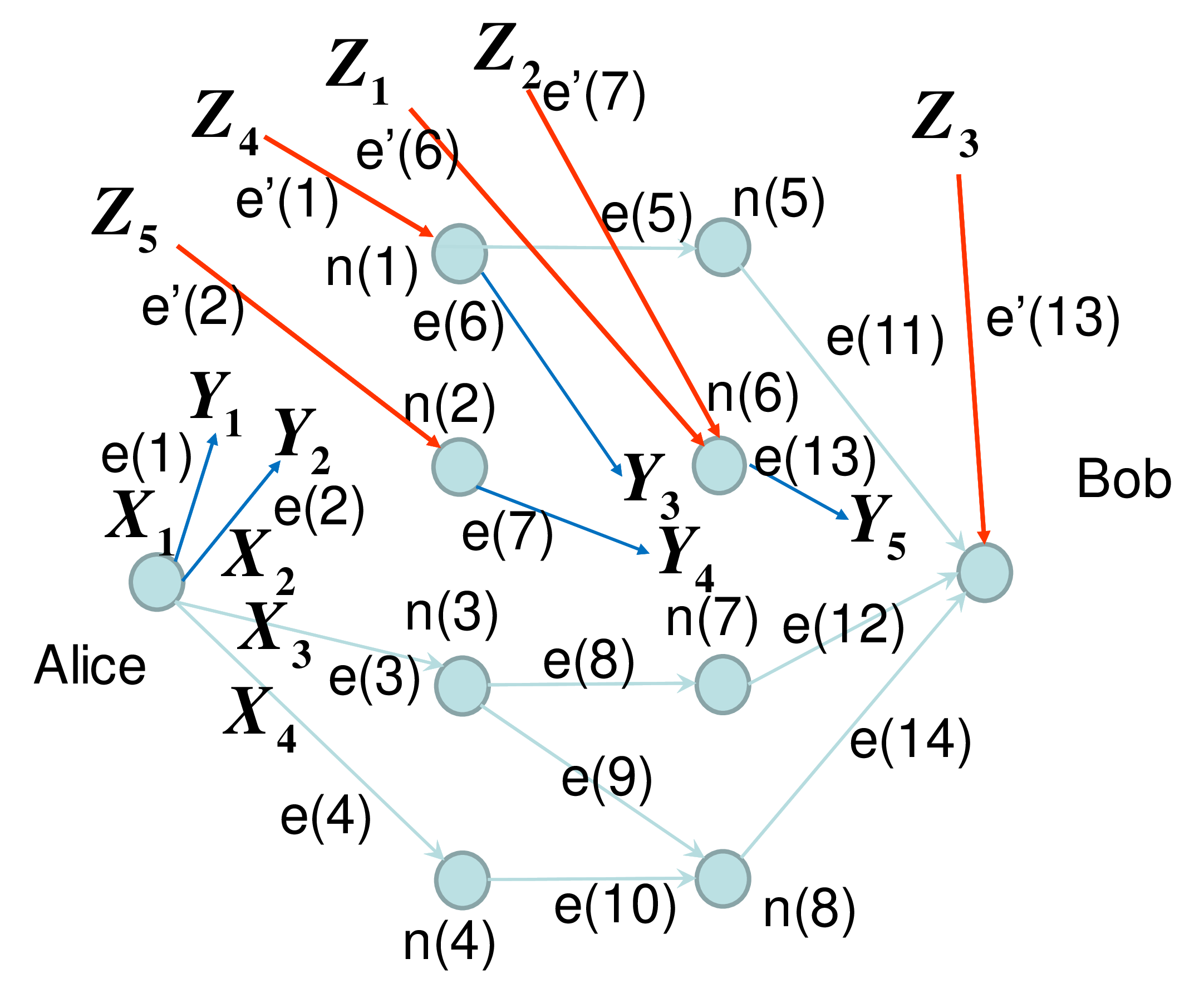}
\end{center}
\caption{Network of Subsubsection \ref{F1Ex} with wiretap and replacement attack.
The edges $e'(1)$, $e'(2)$, $e'(6)$, $e'(7)$,
and $e'(13)$ are the edges to inject the replaced information.}
\Label{F3B}
\end{figure}%

\newpage

\subsection{Finite-length setting and reduction theorem}\Label{S2-5}
Now, we consider the $n$-transmission setting, where Alice uses the same network $n$ times to send a message to Bob.
Alice's input variable (Eve's added variable) is given as 
a matrix $X^n \in\FF_q^{m_{{3}} \times n}$ (a matrix $Z^n \in\FF_q^{m_{5} \times n}$),
and Bob's (Eve's) received variable is given as 
a matrix $Y_B^n\in\FF_q^{m_{{4}} \times n}$ (a matrix $Y_E^n\in\FF_q^{m_{6} \times n}$).
We assume that the topology and dynamics of the network 
and the edge attacked by Eve
do not change during $n$ transmissions.
Their relation is given as
\begin{align}
Y_B^n&=K_B X^n+ H_B Z^n, \Label{F1n} \\
Y_E^n&=K_E X^n+ H_E Z^n. \Label{F2n}
\end{align}
Then, we assume that Eve's strategy $\alpha^n$
is given as a function from $Y_E^n$ to $Z^n$.
We extend the uniqueness condition to the $n$-transmission version.
\begin{definition}
For any value of $K_E x^n$,
there uniquely exists $y^n \in \FF_q^{m_6 \times n} $ such that
\begin{align}
y^n= K_E x^n+ H_E \alpha^n(y).\Label{Uni2}
\end{align}
This condition is called the {\it uniqueness condition}.
\end{definition}

Here, there are two possibilities to define the time ordering of the transmission.
In the first case, 
while $X^n \in \FF_q^{m_{{3}} \times n}$ is composed of $n$ column vectors,
the transmission of the $i$-th column vector is 
performed in the network after the transmission of the $i-1$-th column vector. 
Hence, Eve is allowed to decide the attack $\alpha^n$ based on the previous memory.
In the second case, $n$ transmissions of the $j$-th edge is performed
after $n$ transmissions of the $j-1$-th edge.
In this way, possible strategies $\alpha^n$ depends on this choice of the time ordering.
However, in both cases, Eve's strategy $\alpha^n$ needs to satisfy the uniqueness condition, which can be shown in the same way as Lemma \ref{LL2}.
Hence, we only impose the uniqueness condition.

We formulate a code to discuss the secrecy.
Let ${\cal M}$ and ${\cal L}$ be the message set and the set of values of the scramble random number,
which is often called the private randomness. 
Then, an encoder is given as a function $\phi_n$ from ${\cal M} \times {\cal L} $
to $\FF_q^{m_{{3}} \times n}$, and the decoder is given as $\psi_n$ from $\FF_q^{m_{{4}} \times n}$ to ${\cal M}$.
That is, the decoder does not use the scramble random number $L$ because 
it is not shared with the decoder. 
Our code is the pair $(\phi_n,\psi_n)$, and is denoted by $\Phi_n$.
Then, we denote the message and  the scramble random number as $M$ and $L$. 
The cardinality of ${\cal M}$ is called the size of the code and is denoted by $|\Phi_n|$.
More generally, when we focus on a sequence $\{l_n\}$ instead of $\{n\}$,
an encoder $\phi_n$ is a function from ${\cal M} \times {\cal L} $
to $\FF_q^{m_{{3}} \times l_n}$, and the decoder $\psi_n$ is a function from $\FF_q^{m_{{4}} \times l_n}$ to ${\cal M}$.

Here, we treat $K_B,K_E,H_B$, and $H_E$ as deterministic values, and denote the pairs $(K_B,K_E)$ and $(H_B,H_E)$ by $\bm{K}$ and $\bm{H}$, respectively.
Also, we assume that the matrices $\bm{K}$ and $\bm{H}$ are not changed during transmission.
In the following, we fix $\Phi_n,\bm{K},\bm{H},\alpha^n$.
As a measure of the leaked information, we adopt
the mutual information $I(M; Y_E^n,Z^n)$ between $M$ and Eve's information $Y_E^n$ and $Z^n$.
Since the variable $Z^n$ is given as a function of $Y_E^n$,
we have $I(M; Y_E^n,Z^n)=I(M; Y_E^n)$.
Since the leaked information is given as a function of 
$\Phi_n,\bm{K},\bm{H},\alpha^n$ in this situation, 
we denote it by $I(M;Y_E^n)[\Phi_n,\bm{K},\bm{H},\alpha^n]$.
If we always choose $Z^n=0$, the attack is the same as the passive attack.
This strategy is denoted by $\alpha^n=0$.
When $\bm{K},\bm{H}$ are treated as random variables independent of $M,L$,
the leaked information is given as the expectation of $I(M;Y_E^n)[\Phi_n,\bm{K},\bm{H},\alpha^n]$.
This probabilistic setting expresses the following situation. 
Eve cannot necessarily choose edges to be attacked by herself.
But she knows the positions of the attacked edges,
and chooses her strategy depending on the attacked edges.

\begin{remark}
It is better to remark that there are two kinds of formulations in network coding
even when the network has only one sender and one receiver.
Many papers \cite{Cai2002,Cai06a,Cai06,CN11,CN11b} adopt the formulation, where
the users can control the coding operation in intermediate nodes.
However, this paper adopts another formulation, in which,
the users can control the coding operation only for the input variable $X$ and the output variable $Y_B$ like the papers \cite{JLKHKM,Jaggi2008,JL,Yao2014,KMU,Zhang}.
In the former setting, it is often allowed to employ the private randomness in intermediate nodes.
However, in our setting, 
since no coding operation is allowed in intermediate nodes,
the private randomness is not employed in intermediate nodes.
Remember that the operations in intermediate nodes are linear and are not changed during transmission.
Here, we define the coding operation is considered to be an operation across several alphabets.

In addition, any linear operation over the vector space $\FF_q^{m_3 n}$
is allowed for the encoding process,
our formulation can be regarded as vector linearity. 
\end{remark}

Now, we have the following reduction theorem.
\begin{theorem}[Reduction Theorem]\Label{T1}
When the triplet $(\bm{K},\bm{H},\alpha^n)$
satisfies the uniqueness condition, 
Eve's information $Y_E^n(\alpha^n)$ with strategy $\alpha^n$ can be calculated from 
Eve's information $Y_E^n(0)$ with strategy $0$ (the passive attack),
and $Y_E^n(0)$ is also calculated from $Y_E^n(\alpha^n)$.
Hence, we have the equation
\begin{align}
I(M;Y_E^n)[\Phi_n,\bm{K},0,0]
=
I(M;Y_E^n)[\Phi_n,\bm{K},\bm{H},0]
=
I(M;Y_E^n)[\Phi_n,\bm{K},\bm{H},\alpha^n]. \Label{NCT}
\end{align}
\end{theorem}

This theorem shows that the information leakage of the active attack with the strategy $\alpha^n$ 
is the same as 
the information leakage of the passive attack.
Hence, to guarantee the secrecy under an arbitrary active attack,
it is sufficient to show secrecy under the passive attack.

\begin{proof}
Since the first equation follows from the definition, 
we show the second equation.
We define two random variables $Y_E^n(0):=K_E X^n$ and 
$Y_E^n(\alpha^n):=K_E X^n+ H_E Z^n$.
Due to the uniqueness condition of $Y_E^n(\alpha^n)$, 
for each $Y_E^n(0)=K_E X^n$,
we can uniquely identify 
$Y_E^n(\alpha^n)$.
Therefore, we have 
$I(M;Y_E^n(0) )  \ge I(M;Y_E^n(\alpha^n))$.
Conversely, since $Y_E^n(0)$ is given as a function of $Y_E^n(\alpha^n)$, $Z^n$, and $H_E$,
we have the opposite inequality.
\end{proof}

\begin{remark}[Number of choices]
To compare passive and active attacks,
we count the number of choices of both attacks.
While the passive attack is characterized by the matrix $K_E$,
the information leaked to Eve in the passive attack depends only on the kernel of the matrix $K_E$.
To characterize the information leaked to Eve,
we consider two matrices to be equivalent when their kernels are the same.
In a passive attack, when we fix the rank of $K_E$ (the dimension of leaked information),
by taking into account the equivalent class,
the number of possible choices is upper bounded by
$q^{m_{2}(m_{{3}}-m_{2})} $. 
With an active attack, this calculation is more complicated.
For simplicity, we consider the case with $n=1$.
To consider the minimum number of choices of $\alpha^n$,
we assume condition \eqref{F14}.
(When we do not make this assumption, the number of choices is larger.)
We do not count the choice for the inputs on the edge $\eta(i)$ with $\eta(i) \ge \zeta(m_{6})$ because it does not affect Eve's information.
Then, even when we fix the matrices $K_E,H_E$,
the number of choices of $\alpha^n$ is 
\begin{align}
q^{\sum_{i: \eta(i) < \zeta(m_{6})} q^{T_i}}, \Label{F17}
\end{align}
where $T_i:= \max \{j | \eta(i) \ge \zeta(j)\}$.
Notice that $T_i=i$ when $E_A=E_E$.
If we count the choice on the remaining edges, we need to multiply $q^{\sum_{i: \eta(i) \ge \zeta(m_{6})} q^{T_i}}$ on \eqref{F17}.
For a generic natural number $n$,
the number of choices of $\alpha^n$ is
\begin{align}
 q^{n\sum_{i: \eta(i) < \zeta(m_{6})} q^{n T_i}} \Label{F18}.
\end{align}
\end{remark}

\begin{remark}
Theorem \ref{T1} discusses the unicast case.
It can be trivially extended to the multicast case
because we do not discuss the decoder.
It can also be extended to the multiple unicast case, 
whose network is composed of several pairs of sender and receiver.
When there are $k$ pairs in this setting, the messages $M$ and the scramble random numbers $L$ have the forms
$(M_1, \ldots, M_k)$ and $(L_1, \ldots, L_k)$.
Thus, we can apply Theorem \ref{T1} to the multiple unicast case.
\end{remark}

\begin{remark}
One may consider that 
Theorem \ref{T1} requires the acyclic condition for the network.
However, this condition is not needed because the statement of this theorem follows from 
the uniqueness condition.
That is, when a cyclic network does not satisfy the condition \eqref{NCT},
it does not satisfy the uniqueness condition.
\end{remark}

\begin{remark}
One may consider the following type of attack for an integer $n$ when Alice sends the $i$-th transmission after 
Bob receives the $i-1$-th transmission.
Eve changes the edge to be attacked in the $i$-th transmission
dependently of the information that Eve obtains in the previous $i-1$ transmissions.
Such an attack was discussed in \cite{SHIOJI} when there is no noise injection.
Theorem \ref{T1} does not consider such a situation because it assumes that Eve attacks the same edges
for each transmission to make consistency with the latter sections (Sections \ref{S3} and \ref{S5}).
However, Theorem \ref{T1} can be applied to this kind of attack in the following way.
That is, we find that Eve's information with noise injection can be simulated by 
Eve's information without noise injection even when the attacked edges are changed in the above way.
When we have $n$ transmission over the graph $(V,E)$, 
we consider the graph $(V_n,E_n)$, where
$V_n:=\{(v,i)\}_{v\in V, 1\le i \le n}$ and $E_n:=\{(e,i)\}_{e\in E, 1\le i \le n}$
and $(v,i)$ and $(e,i)$ express the vertex $v$ and the edge $e$ on the $i$-th transmission, respectively.
Hence, when we apply Theorem \ref{T1} with $n=1$ to the graph $(V_n,E_n)$, we obtain the above statement.
\end{remark}

\section{Asymptotic setting with secrecy and robustness}\label{S3}
Next, under the same assumption as that in Section \ref{S2-5},
we consider the asymptotic setting by taking account of robustness as well as secrecy
while Eve's strategy $\alpha$ is assumed to satisfy the uniqueness condition.
We previously assumed that 
the matrices $K_B$, $K_E$, $H_B$, and $H_E$, i.e., 
the topology and dynamics of the network and the edge attacked by Eve do not change during $n$ transmissions.
Now, we assume that Eve knows these matrices and 
that Alice and Bob know none of them because Alice and Bob often do not know the topology and/nor dynamics of the network and/nor the places of the edges attacked by Eve.
However, due to the limitation of Eve's ability,
we assume that the dimension of the information leaked to Eve and 
the rank of the information injected by Eve are limited to $m_2 $ and $m_1$, respectively.
Indeed, when the original network is given by the graph $(V,E)$ and
Eve eavesdrops at most $m_6'$ edges and injects the noise at most $m_5'$ edges,
we have $m_2 \le m_6'$ and $m_1\le m_5'$.
This evaluation is till valid even in the wiretap and replacement model.
Therefore, it is natural to assume the upper bounds of these dimensions.
(See Remark \ref{R-Dim}.)

When Eve adds the error $Z^n$, there is a possibility that Bob cannot recover the original information $M$.
This problem is called {\it robustness}, and may be regarded as a kind of error correction.
Under the conventional error correction, 
the error $Z^n$ is treated as a random variable subject to a certain distribution.
However, our problem 
is different from the conventional error correction 
because the decoding error probability depends on the strategy $\alpha^n$.
Hence, we denote it by $P_e[\Phi_n,\bm{K},\bm{H},\alpha^n]$.
Then, the following proposition is known.

\begin{proposition}[\protect{\cite{JLKHKM,Jaggi2008,JL,Yao2014}}]\Label{T2}
Assume that 
$m_{2}+m_{1}< m_{0}$.
There exists a sequence of codes $\Phi_{n}$ of block-length $l_n$ on a finite field $\FF_q$
whose message set is $\FF_{q}^{k_{n}}$ such that
\begin{align}
&\lim_{n \to \infty} \frac{k_n}{l_n} = m_{0}-m_{1}\\
&\lim_{n \to \infty} \max_{ \bm{K},\bm{H}} \max_{\alpha^n} 
P_e[\Phi_n, \bm{K},\bm{H},\alpha^n]=0,
\end{align}
where the maximum is taken with respect to
$(K_B,H_B,K_E,H_E) \in 
\FF_q^{m_{{4}} \times m_{{3}}}\times 
\FF_q^{m_{{4}} \times m_{5}}\times 
\FF_q^{m_{6} \times m_{{3}}}\times 
\FF_q^{m_{6} \times m_{5}}$ with \eqref{1-6} and 
\eqref{1-6X}.
Here, there is no restriction for the choice of $m_5$ and $m_6$.
\end{proposition}

The existing proof of Proposition \ref{T2}
is given as a combination of several results.
Each part of the existing proof is hard to read because it omits the detail derivation. 
Hence, for readers' convenience, 
we give its alternative proof in Appendix \ref{Ap1}, which has an improvement over the existing proof. 
Combining Theorem \ref{T1} and Proposition \ref{T2}, we obtain the following theorem.
\begin{theorem}\Label{T3}
We assume that
$m_{2}+m_{1}< m_{0}$.
There exists a sequence of codes $\Phi_{n}$ 
of block-length $l_n$ on finite field $\FF_q$
whose message set is $\FF_q^{k_n}$ such that
\begin{align}
&\lim_{n \to \infty} \frac{k_n}{l_n} = m_{0}-m_{1}-m_{2}\\
&\lim_{n \to \infty} \max_{ \bm{K},\bm{H}} \max_{\alpha^n} 
P_e[\Phi_n,\bm{K},\bm{H},\alpha^n]=0 \Label{H3-181} \\
&
\max_{ \bm{K},\bm{H}} \max_{\alpha^n} 
I(M;Y_E^n)[\Phi_n,\bm{K},\bm{H},\alpha^n]=0, \Label{H3-182}
\end{align}
where the maximum is taken in the same way as with Proposition \ref{T2}.
\end{theorem}

Before our proof, we prepare basic facts about information-theoretic security. 
We focus on a random hash function $f_R$ from ${\cal X}$ to ${\cal Y}$ with random variable $R$ deciding the function $f_R$.
It is called {\it universal2} when 
\begin{align}
{\rm Pr}\{ f_R(x)=y\}\le \frac{|{\cal Y}|}{|{\cal X}|}
\end{align}
for any $x \in {\cal X}$ and $y \in {\cal Y}$.

For $s\in (0,1]$, we define the conditional R\'{e}nyi entropy $H_{1+s}(X|Z)$ for the joint distribution $P_{XZ}$ as \cite{hayashi11}
\begin{align}
H_{1+s}(X|Z):= \frac{-1}{s}\log \sum_{z \in {\cal Z}} P_Z(z) \sum_{x \in {\cal X}} P_{X|Z}(x|z)^{1+s},
\end{align}
which is often denoted by $H_{1+s}^{\uparrow}(X|Z)$ in \cite{TBH,HW}.
When $X$ obeys the uniform distribution, we have 
\begin{align}
H_{1+s}(X|Z) \ge  \log\frac{  |{\cal X}|}{|{\cal Z}|}. \Label{F15}
\end{align}

\begin{proposition}
\cite{bennett95privacy,HILL}\cite[Theorem 1]{hayashi11}\Label{T4}
\begin{align}
I(f_R(X);Z|R) \le \frac{e^{s \log |{\cal Y}|- H_{1+s}(X|Z)}}{s}
\end{align}
\end{proposition}
for $s \in (0,1]$.

\begin{proofof}{Theorem \ref{T3}}
We choose a sequence of codes $\{\Phi_n=(\phi_n,\psi_n)\}$  given in Corollary \ref{T2}.
We fix $ \bar{k}_n:= k_n- m_{2} {l_n}- \lceil \sqrt{l_n}\rceil$.
Now, we choose a universal2 linear surjective random hash function $f_R$ from $\FF_q^{k_n}$ to $\FF_q^{\bar{k}_n}$.

To construct our code, we consider a virtual protocol as follows. 
First, Alice sends a larger message $M$ by using the code $\Phi_n$,
and Bob recovers it.
Second, Alice randomly chooses $R$ deciding the hash function $f_R$ and sends it to Bob via a public channel.
Finally, Alice and Bob apply the hash function $f_R$ to their message, and denote the  result value by $\bar{M}$ so that Alice and Bob share the information  $\bar{M}$ with a probability of close to $1$.

Since the conditional mutual information between
$\bar{M}$ and $Y_E^{l_n}$ depends on $\Phi_n,\bm{K},\bm{H},\alpha^n$,
we denote it by $ I(\bar{M};Y_E^{l_n}|R)[\Phi_n,\bm{K},\bm{H},\alpha^n]$.
Theorem \ref{T1} shows $ I(\bar{M};Y_E^{l_n}|R)[\Phi_n,\bm{K},\bm{H},\alpha^n]=
 I(\bar{M};Y_E^{l_n}|R)[\Phi_n,\bm{K},\bm{H},0]$, which does not depend on $K_B,H_B ,H_E$ and depends only on $K_E$.
Now, we evaluate this leaked information via a similar idea to that reported in \cite{Matsumoto2011}.
Since inequality \eqref{F15} implies that 
$H_{1+s}(M|Y_E^{l_n}) \ge( k_n -{l_n} m_{2}) \log q$,
Proposition \ref{T4} yields
\begin{align}
& I(\bar{M};Y_E^{l_n}|R)[\Phi_n,\bm{K},\bm{H},0]
\le \frac{e^{s (\bar{k}_n\log q - H_{1+s}(M|Y_E^{l_n}))}}{s} \nonumber \\
\le & \frac{q^{s (\bar{k}_n- k_n +{l_n} m_{2})}}{s}
\le \frac{q^{-s \lceil \sqrt{l_n}\rceil }}{s}.
\end{align}
We set $s=1$.
For each matrix $K_E \in \FF_q^{m_{6}\times m_{{3}}}$ satisfying $\rank K_E= m_2$,
Markov inequality guarantees that the inequality
\begin{align}
I(\bar{M};Y_E^{l_n})[\Phi_n,\bm{K},\bm{H},0]|_{R=r}
\le q^{- \lceil \sqrt{l_n}\rceil +c+1} \Label{F20a}
\end{align}
holds at least with probability $1-q^{-c-1}$.
Since the number of matrices $K_E$ satisfying $\rank K_E=m_{2} $ is upper bounded by $ q^{m_{6}m_{{3}}}$,
there exists a matrix $K_E \in \FF_q^{m_{6}\times m_{{3}}}$ such that
$\rank K_E= m_2$ and \eqref{F20a} does not hold
at most with probability $q^{m_{6}m_{{3}}} q^{-c-1}$.
Hence, \eqref{F20a} holds for any matrix $K_E \in \FF_q^{m_{6}\times m_{{3}}}$ satisfying $\rank K_E= m_2$
at least with probability $1-q^{m_{6}m_{{3}}} q^{-c-1}$.
Letting $c$ be $m_{6}m_{3}$, we have
\begin{align}
I(\bar{M};Y_E^{l_n})[\Phi_n,\bm{K},\bm{H},0]|_{R=r}
\le q^{- \lceil \sqrt{l_n}\rceil +m_{6}m_{{3}}+1} \Label{F20}
\end{align}
for any matrix $K_E \in \FF_q^{m_{6}\times m_{{3}}}$ satisfying $\rank K_E= m_2$
at least with probability $1-\frac{1}{q}$.
Therefore, there exists a suitable hash function $f_r$ such that
\begin{align*}
I(\bar{M};Y_E^{l_n})[\Phi_n,\bm{K},\bm{H},0]|_{R=r}
\le q^{- \lceil \sqrt{l_n}\rceil +m_{6}m_{{3}}+1},
\end{align*}
which goes to zero as $n$ goes to infinity because 
$m_{6}m_{{3}}+1$ is a constant.
Since the code is linear,
Eve observes a subspace of input information
$\FF_q^{\bar{k}_n} $.
Hence, the amount of leaked information is an integer times of
$\log q$.
Hence, as discussed in  \cite{Matsumoto2011a},
when $l_n$ is sufficiently large, there exists a suitable hash function $f_r$ such that
\begin{align*}
I(\bar{M};Y_E^{l_n})[\Phi_n,\bm{K},\bm{H},0]|_{R=r}
=0.
\end{align*}

Now, we return to the construction of real codes.
We choose the sets ${\cal M}$ and ${\cal L}$ as
$\FF_q^{\bar{k}_n}$ and $\FF_q^{m_{2} {l_n}+ \lceil \sqrt{l_n}\rceil}$, respectively.
Since the linearity and the surjectivity of $f_r$ implies that
$|f_r^{-1}(x)|=q^{m_{2} {l_n}+ \lceil \sqrt{l_n}\rceil}$ for any element $x\in {\cal M} $,
we can define the invertible function $\bar{f}_r$ from ${\cal M}\times {\cal L}$
to the domain of $f_r$, i.e., $\FF_q^{k_n} $ such that $\bar{f}_r^{-1} (f_r^{-1}(x))=\{x\}\times {\cal L}$
for any element $x\in {\cal M} $.
This condition implies that $f_r \circ \bar{f}_r(x,y)=x$ for $(x,y)\in {\cal M}\times {\cal L}$.
Then, we define our encoder as $\bar{\phi}_n:= \phi_n \circ \bar{f}_r$,
and our decoder as $\bar{\psi}_n:= f_r\circ \psi_n $.
The sequence of codes $\{(\bar{\phi}_n,\bar{\psi}_n)\}$
satisfies the desired requirements.
\end{proofof}

\begin{remark}\Label{R-Dim}
If we replace the condition \eqref{1-6X} by the condition
\begin{align}
\rank H_B\le m_{1}, ~ \rank K_E\le m_{2} ,
\Label{1-6X2}
\end{align}
the Proposition \ref{T2} and Theorem \ref{T3} still hold due to the following reason.
For $(K_B,H_B,K_E,H_E)$ to satisfy \eqref{1-6} and \eqref{1-6X2}, 
there exists $(K_B',H_B',K_E',H_E')$ to satisfy \eqref{1-6} and \eqref{1-6X} such that 
$P_e[\Phi_n, \bm{K},\bm{H},\alpha^n] \le P_e[\Phi_n, \bm{K}',\bm{H}',\alpha^n]$
and
$I(M;Y_E^n)[\Phi_n,\bm{K},\bm{H},\alpha^n]\le
I(M;Y_E^n)[\Phi_n,\bm{K}',\bm{H}',\alpha^n]$.
Hence,
the Proposition \ref{T2} and Theorem \ref{T3} still hold under this modification.
\end{remark}

\begin{remark}[Efficient code construction]
We discuss an efficient construction of our code from a code 
$(\phi_n,\psi_n)$ given in Corollary \ref{T2} with $q=2$.
A modified form of the Toeplitz matrices is also shown to be 
a universal2 linear surjective hash function, 
which is given by a concatenation $(T(S), I)$ of the 
$(m_{2} l_n+ \lceil \sqrt{l_n}\rceil) \times \bar{k}_n$ Toeplitz matrix $T(S)$ and 
the $\bar{k}_n \times \bar{k}_n$ identity matrix $I$ \cite{Hayashi-Tsurumaru},
where $S$ is the random seed used to decide the Toeplitz matrix and belongs to $\FF_2^{k_n-1}$.
The (modified) Toeplitz matrices are particularly useful in practice, because there exists an efficient multiplication algorithm using the fast Fourier transform algorithm with complexity $O(l_n\log l_n)$. 

When the random seed $S$ is fixed, 
the encoder for our code is given as follows.
By using the scramble random variable $L \in \FF_2^{m_{2} l_n+ \lceil \sqrt{l_n}\rceil)}$,
the encoder $\bar{\phi}_n$ is given as $\phi_{n}
\Big( 
\Big(
\begin{array}{cc} 
I & - T(S) \\
0 & I
\end{array}
\Big)
\Big(
\begin{array}{c} 
M \\
L
\end{array}
\Big)
\Big)$
because 
$
(I,T(S))
\Big(
\begin{array}{cc} 
I & - T(S) \\
0 & I
\end{array}
\Big)
= (I, 0)$.
(The multiplication of Toeplitz matrix $T(S)$ can be performed as a part of a circulant matrix. 
For example, the reference \cite[Appendix C]{Hayashi-Tsurumaru}
provides a method to give a circulant matrix.).
A more efficient construction for univeral2 hash function is discussed in \cite{Hayashi-Tsurumaru}.
Hence, the decoder $\bar{\psi}_n$ is given as $Y_B^{l_n} \mapsto (I,T(S)) \psi_{n}(Y_B^{l_n})$.
\end{remark}

\begin{remark}
Here, we clarify the difference between our results and the setting of the preceding papers \cite{KMU,Yao2014,Zhang,Rashmi}, which consider correctness and secrecy.
Their secrecy analysis is different from our analysis
although the code construction in \cite{KMU,Yao2014,Zhang} does not depend on the concrete form of 
matrices $K_B, K_E, H_B, H_E$, which is similar to our code construction.

While the papers \cite{SK,Yao2014} considered correctness when the error exists,
it discusses the secrecy only when there is no error.
Similarly, the paper \cite{Rashmi} considers a different active adversary model, in which,
it discusses
the node-repair and data-reconstruction operations
even in the presence of such an attack
while the model of passive eavesdroppers in the paper \cite{Rashmi} 
discusses the secrecy with respect to the message to be transmitted.
Indeed, the papers \cite{SK,Yao2014} provided a statement similar to Theorem \ref{T3}.
However, it showed only Eq. \eqref{H3-181} and 
$\lim_{n \to \infty} \max_{ \bm{K},\bm{H}} 
I(M;Y_E^n)[\Phi_n,\bm{K},\bm{H},0]=0$
instead of \eqref{H3-182}
by combining Proposition \ref{T2B} and the result of the paper \cite{SK}. 
To show \eqref{H3-182}, we need to employ Theorem \ref{T1}.
If we do not apply Theorem \ref{T1} in step \eqref{F20} in our proof of Theorem \ref{T3},
we have to multiply the number of choices of strategy $\alpha^n$.
As a generalization of \eqref{F17}, this number is given in \eqref{F18},
which grows up double-exponentially. 
Hence, our proof of Theorem \ref{T3} does not work without the use of Theorem \ref{T1}.

While the papers \cite[Proposition 5]{KMU}\cite{Zhang} consider the secrecy when the error exists,
it addresses the amount of leaked information only when the eavesdropper does not know the information of the noise.
That is, they evaluate the mutual information between $ M$ and $Y_E^n $.
However, our analysis evaluates the leaked information when the eavesdropper knows the information of the noise.
That is, we address the mutual information between $ M$ and the pair $(Y_E^n,Z^n) $.
\end{remark}

\section{Asymptotic setting with secrecy}\Label{S5}
Next, we consider the case when only the secrecy is imposed and the robustness is not imposed.
In this case, we impose the correctness of the case with passive attack instead of the robustness.
That is, we impose the following condition.
\begin{align}
\lim_{n \to \infty} \max_{ \bm{K}} P_e[\Phi_n, \bm{K},0,0]
=\lim_{n \to \infty} \max_{ \bm{K},\bm{H}} P_e[\Phi_n, \bm{K},\bm{H},0]=0.
\Label{H3-18K}
\end{align}
Here, as the secrecy, we impose the following condition.
\begin{align}
 \max_{ \bm{K},\bm{H}} \max_{\alpha^n} 
I(M;Y_E^n)[\Phi_n,\bm{K},\bm{H},\alpha^n]=0\Label{47-2} .
\end{align}
Here, both maximums are taken with respect to
$(K_B,K_E) \in 
\FF_q^{m_{{4}} \times m_{{3}}}\times 
\FF_q^{m_{{2}} \times m_{3}}$ with \eqref{1-6}.
We notice that the situation of the correctness \eqref{H3-18K}
is different from the situation of the secrecy \eqref{47-2}.
The correctness \eqref{H3-18K} addresses only the case with passive attack,
but the secrecy \eqref{47-2} addresses the cases with active attack.
The following is the reason why we consider this setting.

Consider that Alice and Bob can communicate with each other by using a public channel, which allows Alice and Bob to communicate with each other without any error, but
the secrecy is not guaranteed.
In this case, when Alice and Bob share a sufficient number of secret random variables,
they can communicate with each other securely.
To share such secret random variables,
they can send them via the secure network coding.
Now, we consider this problem in an asymptotic setting, where
the secrecy condition \eqref{H3-182} is definitely required.
However, robustness \eqref{H3-181} is not necessary 
because they can check whether or not the transmitted random number is correct 
by using an error verification test with a public channel after the transmission \cite[Section VIII]{Fung} \cite[Step 4 of Protocol 2]{H17}.
Hence, due to the error verification,
it is sufficient to impose condition \eqref{H3-18K} instead of \eqref{H3-181}.
Indeed, it is not easy to check whether $\bm{H}$ and $\alpha^n$ are $0$
even when the error verification test is passed
because there is a possibility that the error caused by $\bm{H}$ and $\alpha^n$ 
can be corrected by the code $\Phi_n$.
Since we cannot ignore the possibility that $\bm{H}$ and $\alpha^n$ are not $0$,
we cannot relax the secrecy condition \eqref{47-2}.
This setting appears when we consider quantum key distribution, as explained in Section \ref{S6}.
We use the following theorem to analyze this problem.

\begin{theorem}\Label{T3B}
There exists a sequence of codes $\Phi_{n}$ 
of block-length $l_n$ on finite field $\FF_q$
whose message set is $\FF_q^{k_n}$ such that conditions \eqref{H3-18K} and \eqref{47-2} and
\begin{align}
&\lim_{n \to \infty} \frac{k_n}{l_n} = m_{0}-m_{2} \Label{47-1}
\end{align}
holds.
\end{theorem}

From the definition, we see that
$P_e[\Phi_n, \bm{K},0,0]=P_e[\Phi_n, \bm{K},\bm{H},0]$.
Also, note that $P_e[\Phi_n, \bm{K},0,\alpha^n]$ does not depend on $K_E$.
Further, the rate $m_{0}-m_{2}$ is asymptotically optimal, 
i.e., there is no code surpassing the rate $m_{0}-m_{2}$,
which follows from the converse part of the conventional wire-tap channel \cite{wyner75,csiszar78}.

To show the above theorem, as a special case of Theorem \ref{T3} with $m_1=0$, 
we prepare the following corollary.

\begin{corollary}\Label{T2E}
There exists a sequence of codes $\Phi_{n}$ 
of block-length $l_n$ on finite field $\FF_{q}$
whose message set is $\FF_{q}^{k_{n}}$ such that
\begin{align}
&\lim_{n \to \infty} \frac{k_{n}}{l_n} = m_{0}-m_{2}\\
&
\max_{ \bm{K}}
I(M;Y_E^n)[\Phi_n,\bm{K},0,0]=0,\\
&\lim_{n \to \infty} \max_{ \bm{K}} 
P_e[\Phi_n, \bm{K},0,0]=0,
\end{align}
where the maximum is taken with respect to
$(K_B,K_E) \in 
\FF_q^{m_{{4}} \times m_{{3}}}\times 
\FF_q^{m_{{2}} \times m_{3}}$ under the condition \eqref{1-6}.
\end{corollary}

Combining Corollary \ref{T2E} and Theorem \ref{T1},
we obtain Theorem \ref{T3B}.

Here, we compare existing results with Corollary \ref{T2E}.
As a similar result to Corollary \ref{T2E}, the following proposition is known.
Since Corollary \ref{T2E} does not require the assumptions $m_3=m_4=m_0 $ and $K_B=I$,
Corollary \ref{T2E} is slightly advantageous.
Hence, Theorem \ref{T3B} is a stronger statement than the following existing statement.

\begin{proposition}[\protect{\cite[Theorem 7]{Matsumoto2011a}},\cite{KMU}]\Label{T2D}
We assume that $m_3=m_4=m_0 $ and $K_B=I$.
There exists a sequence of codes $\Phi_{n}$ 
of block-length $n$ on finite field $\FF_{q}$
whose message set is $\FF_{q}^{k_{n}}$ such that
\begin{align}
&\lim_{n \to \infty} \frac{k_{n}}{n} = m_{0}-m_{2}\\
&\lim_{n \to \infty} \max_{ \bm{K}}
I(M;Y_E^n)[\Phi_n,\bm{K},0,0]=0,\\
&\lim_{n \to \infty} 
P_e[\Phi_n, \bm{K},0,0]=0,
\end{align}
where the maximum is taken with respect to
$K_E \in \FF_q^{m_{{2}} \times m_{0}}$.
\end{proposition}


\section{Application to network quantum key distribution}\Label{S6}
In this section, to realize long distance communication with quantum key distribution,
using the result in Section \ref{S5},
we consider a network of quantum key distribution as follows.
The authorized sender, Alice, is connected to the authorized receiver, Bob,
via the network given by the graph $({V}, {E})$ with $|E|=k$.
A linear operation is fixed in each node so that 
we have the relation $Y_B=K_B X$ with Alice's input $X$ and Bob's output $Y_B$.
Then, if secure information transmission is available on each edge,
secure communication from Alice to Bob can be realized.
For every edge $(u,v)\in E$,
the distant nodes $u$ and $v$ generate secure common keys by quantum key distribution.
That is, $k$ pairs of secure keys are generated by quantum key distribution.
In the following, 
we discuss how we can make secure message transmission from Alice to Bob
by using these $k$ pairs of secure keys with public channels.
This kind of secure communication is called network quantum key distribution.

First, we consider the case when all nodes are authenticated.
In this case, Alice can securely send her message $X$ to Bob in the following way.
Let $X_i$ be the random variable to be transmitted on the $i$-th edge.
Let $Z_i$ be the secure keys generated in the $i$-th edge by quantum key distribution.
When $X_i$ is directly transmitted, this information transmission is not secure.
To realize security, $Y_i:=X_i+Z_i$ is transmitted on the $i$-th edge, instead.
Then, a secure transmission in each edge is realized.
Hence, due to the above relation $Y_B=K_B X$, 
secure communication from Alice to Bob can be realized.

However, it is very difficult to guarantee security when a part of nodes are occupied by Eve.
Such a model is often called a node adversary model  
while the model introduced in Section \ref{S2} is called an edge adversary model.
The main problem with network quantum key distribution is 
realization of secure communication from Alice to Bob 
under a node adversary model. 
To investigate the security in the node adversary model, 
we convert a given node adversary model to a special case of the edge adversary model as in \cite{TJBK}.
In an edge adversary model,
Eve wiretaps and contaminates the information only on the edges $E_E$.
To apply the model to the current situation, we consider that 
all the edges linked to the nodes occupied by Eve are wiretapped and controlled by Eve.
When these occupied nodes communicate each other, Eve's attack is active attack.
That is, analysis for active attack is essential.
Therefore, we can apply Theorem \ref{T3} to the security analysis of 
the direct transmission of secret message via network quantum key distribution.
In quantum key distribution, it is usual to assume that Alice and Bob share 
secure random numbers whose lengths are asymptotically negligible in block-length $n$
because the asymptotically negligible keys are needed for 
authentication for the public channel.
In this case, to generate secure keys with length $O(n)$, 
we can employ Theorem \ref{T3B},
where the asymptotically negligible keys are used for an error verification test.

For example, we consider the network given in Fig. \ref{F1B}, which has nodes
$v(1), \ldots, v(8)$ as intermediates nodes.
Fig. \ref{F3} expresses the information on each edge in this  network.
This network connects Alice and Bob with rank 4.
The ranks of $K_E$ and $H_B$ of typical cases
are summarized in Table \ref{TB2}.

\begin{table}[htb]
  \begin{center}
    \caption{Ranks}\Label{TB2}
    \begin{tabular}{|c|c|c|} \hline
Nodes  & rank $K_E$ & rank $H_B$ \\ \hline 
$v(1)$ & 1 & 2 \\ \hline 
$v(2)$ & 1 & 1 \\ \hline 
$v(6)$ & 2 & 1 \\ \hline 
$v(2) \& v(5)$ & 2 & 2 \\ \hline 
$v(2) \& v(6)$ & 2 & 1 \\ \hline 
$v(1) \& v(3)$ & 2 & 4 \\ \hline 
$v(1) \& v(2)$ & 2 & 3 \\ \hline 
$v(1) \& v(6)$ & 2 & 2 \\ \hline 
$v(1) \& v(8)$ & 3 & 3 \\ \hline 
$v(6) \& v(8)$ & 4 & 2 \\ \hline 
    \end{tabular}
  \end{center}
\end{table}

When the number of nodes occupied by Eve is limited to 1,
the ranks of $K_E$ and $H_B$ are upper bounded by $2$.
In the latter case, Theorem \ref{T3B} guarantees that
Alice can securely transmit a random number with rank 2 per single use of the network.
In the former case, since $4-2-2=0$, Theorem \ref{T3} cannot guarantee that Alice  securely transmits her message to Bob.

As another example, we consider the network given in Fig. \ref{F4}, in which, 
the nodes connect the next nodes and the nodes after the next.
Assume that we have pairs of secret keys in the network of Fig. \ref{F4}.
We suppose that $v(1)$ intends to communicate with $v(8)$ securely.
They make the network as
$v(1) \to v(12)\to v(10)\to v(8)$,
$v(1) \to v(11)\to v(9)\to v(8)$,
$v(1) \to v(3)\to v(5)\to v(7)\to v(8)$,
$v(1) \to v(2)\to v(4)\to v(6)\to v(8)$, which connects $v(1)$ and $v(8)$ with rank 4.
When Eve occupies one intermediate node,
the ranks of $K_E$ and $H_B$ are one.
In the latter case, Theorem \ref{T3B} guarantees that
Alice in $v(1)$ can securely transmit a random number with rank 3 per single use of the network.
In the former case, 
Theorem \ref{T3} guarantees that Alice in $v(1)$ securely transmits her message to Bob with rank 2 per single use of the network.

\begin{figure}[htbp]
\begin{center}
\includegraphics[scale=0.37]{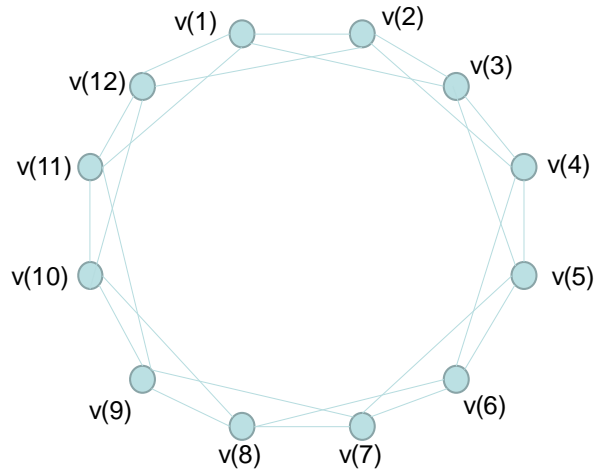}
\end{center}
\caption{Cyclic Network}
\Label{F4}
\end{figure}%

When Eve occupies two intermediate nodes, the ranks of $K_E$ and $H_B$ are at most two.
In the latter case, Theorem \ref{T3B} guarantees that
Alice in $v(1)$ can securely transmit a random number with rank 2 per single use of the network.
In the former case, 
Theorem \ref{T3} cannot guarantee that Alice in $v(1)$ securely transmits her message to Bob.
This method can be generalized to the case when Alice and Bob are $v(i)$ and $v(j)$
with $|i-j|\ge 2 (\bmod~ 12)$.

Indeed, this idea can be generalized to the cyclic network even when the number of nodes is odd. 
Further, this network can be generalized to the following network of quantum key distribution with two integers $k>l>0$.
The set of nodes is given as $\{v(i)\}_{i=1}^k$,
and the set of edges is given as $\{(v(i), v(j)) \}_{|i-j| \le l (\bmod~ k)}$.
Now, we set Alice and Bob as $v(i)$ and $v(j)$ with $|i-j|\ge l (\bmod~ k)$.
Then, they can make $2l$ paths connecting $v(i)$ and $v(j)$ without duplication in the intermediate nodes.
Then, even when $2l-1$ nodes are occupied by Eve,
Alice and Bob can securely share a secret random number due to Theorem \ref{T3B}.

\section{Application to multiple multicast network} \Label{S10}
We consider how to apply our result to
a multiple multicast network with $a$ senders and $ \sum_{i=1}^a b_i$ receivers, in which,
the senders and the receivers are labeled as $i$ and $(i,j)$ with $1\le i\le a$ and $1\le j \le b_i $, respectively.
Sender $i$ intends to securely sends the message to Receiver $(i,j)$.
That is, Sender $i$ wants to keep secrecy for Receiver $(i',j)$ with $i'\neq i$.
In the one-time use of the network, 
Sender $i$ sends $m_{3,i}$ symbols $\vec{X}_i$ of $\FF_q$ via $m_{3,i}$ channels
and
Receiver $(i,j)$ receives $m_{4,i,j}$ symbols $\vec{Y}_{i,j}$of $\FF_q$ via $m_{4,i,j}$ channels.
Without loss of generality, we can assume that 
$m_{4,i,j}$ does not depend on $j$ and is simplified to $m_{4,i}$ due to the following reason.
When $m_{4,i,j} < \max_{j'} m_{4,i,j'}$,
we can consider that Receiver $(i,j)$ receives symbol $0$ via $ \max_{j'} m_{4,i,j'}- m_{4,i,j} $ channels.
If the codes in network is designed perfectly, we have no cross line nor no information leakage to unintended receivers.
In this section, we consider the case with small amount of cross line and information leakage to unintended receivers due to errors on the design of network.

We assume that these senders and receivers are connected via network composed of linear operations.
Then, using matrices $K_{i,j;i'}$, we can describe their relations as
\begin{align}
\vec{Y}_{i,j}= \sum_{i'=1}^{a} K_{i,j;i'} \vec{X}_{i'}.
\end{align}
While the senders transmit their information repeatedly, 
we assume that the coefficient matrices $K_{i,j;i'}$ do not change.
We assume that receivers do not  collude to recover the message from senders.
Now, we apply the model \eqref{E3} to
 the secure communication transmission from Sender $i$ to Receiver $(i,j)$.
When we consider information leakage to Receiver $(i'',j'')$ with $i''\neq i$,
we substitute $\vec{X}_i$, $(\vec{X}_{i'})_{i'\neq i}$, $\vec{Y}_{i,j}$ and $\vec{Y}_{i'',j''}$
into $X$, $Z$, $Y_B$, and $Y_E$, respectively. 
We assume that the rank of information crossed from other senders is $m_{1,i,j}$
and the rank of leaked information to Receiver $(i'',j'')$ is $m_{2,i;i'',j''}$.
We introduce the maximum ranks 
$m_{0,i}:= \max_{j}\rank K_{i,j;i}$, 
$m_{1,i}:= \max_{j}m_{1,i,j}$, and  $m_{2,i} := \max_{i'',j''} m_{2,i;i'',j''}$.
Sender $i$ and Receiver $(i,j)$ are assumed to know only the integers
$m_{0,i}$, $m_{1,i}$, $m_{2,i}$, $m_{3,i}$, $m_{4,i}$ and have no other knowledge for the network structure.
We choose our code by applying Theorem \ref{T2B} to
the case with 
$m_0=m_{0,i}$,
$m_1=m_{1,i}$,
$m_2=m_{2,i}$,
$m_3=m_{3,i}$, and 
$m_4=m_{4,i}$.
Since the code does not depend on the choice of $j$ and $i'',j''$, this code works well in this situation.

\section{Non-linear codes in one hop relay network}\Label{S4}
\subsection{Summary for non-linear codes in the one hop relay network given in Fig. \ref{F1}}\Label{S4-1}
In this section, we focus on the imperfect security, i.e.,
the property that Eve cannot recover the original message with probability one \cite{CN11,CN11b}.
It is known that there exists a linear imperfectly secure code 
over a finite field $\FF_q$ of sufficiently large prime power $q$
when Eve may access a subset of channels that does not contain a cut between Alice and Bob 
even when the linear code does not employ private randomness in the intermediate nodes \cite{Bhattad}.
Theorem \ref{T1} guarantees that 
such a linear code is still imperfectly secure even for active attack over the same network.
However, it is not clear whether there exists such a linear imperfectly secure code over a finite field $\FF_p$ of a prime $p$.
To see how crucial the linearity condition is in Theorem \ref{T1},
we consider this problem over the one hop relay network given in Fig. \ref{F1} with 
edges $E=\{e(1),e(2),e(3),e(4)\}$
only in the single shot setting, i.e., the case when
the integer $n$ defined in Section \ref{S2-5} is $1$, in other words, the sender sends only one element of $\ZZ_d$,
which is called the scalar linearity when $\ZZ_d$ is a finite field $\FF_p$ with prime $p$ \cite[Section I]{DFZ}\footnote{In contrast, 
the linear setting is called the vector linearity when the integer $n$ defined in Section \ref{S2-5} is greater than $1$.
In fact, the paper \cite{Bhattad} discussed this kind of imperfectly secure code
in the case with $n=1$ while it chooses large $q$. 
In contrast, the paper \cite{SK} discussed a similar imperfectly secure code construction
by increasing $n$ (the vector linearity) while it did not increase the size of $q$.
The paper \cite{CCMG} extended this type of vector linearity setting 
 of imperfectly secure codes to the case with multi-source multicast.}.
That is,
we consider the transmission of the message $M$ in $\ZZ_d$ by using the one hop relay network given in Fig. \ref{F1}
when the information at the edges
is given as an element of $\ZZ_d$.
Here, $d$ is an arbitrary natural number, and it is not necessarily a prime number.
In contrast, Theorem \ref{T1} holds with an arbitrary $n$, which is called the vector linearity when $\ZZ_d$ is a finite field $\FF_p$ with prime $p$.
Also, private randomness is allowed in the sender, but
no private randomness is allowed in the intermediate nodes.
Therefore, 
we are allowed to choose an arbitrary deterministic function
$\varphi$ from $\ZZ_d^2$ to $\ZZ_d^2$
as our coding operation on the intermediate node.
Our encoder in the source node is given as a stochastic map $\phi$
from $\ZZ_d$ to $\ZZ_d^2$,
and our decoder is given as a deterministic function
$\psi$ from $\ZZ_d^2$ to $\ZZ_d$.
Eve is allowed to attack two edges of $E$ except for the pairs $\{e(1),e(2)\}$ and $\{e(3),e(4)\}$.
In this section, 
we call the triplet $(\phi,\varphi,\psi)$ a code over the one hop relay network (Fig. \ref{F1}).

We have two attack models, the passive attack and the active attack.
In the passive attack, Eve can eavesdrop two edges, but cannot change the information on the attacked edge.
In the active attack, 
Eve can insert another information on the attacked edge in the first group $\{e(1),e(2)\}$,
and eavesdrop one edge in the second group $\{e(3),e(4)\}$.
That is, we consider the active attack with $E_A=E_E$.
Here, Eve cannot change the edge to be attacked by using the information 
on the attacked edge in the first group $\{e(1),e(2)\}$.
When a code satisfies the following two conditions in the respective models,
the code is called imperfectly secure in the respective models.
Otherwise, it is called insecure in the respective models.
In the following conditions, the information on the edge $e(i)$ is written as $Y_i$.

\begin{description}
\item[(B1)] (Recoverability) 
Bob can recover the message $M$ 
from $Y_3 $ and $Y_4$ when Eve does not make any replacement.

\item[(B2)] (Secrecy)
No active attack $\tilde{\psi}$ from $\ZZ_d^3$ to $\ZZ_d$ 
satisfies one of the following conditions.
\begin{align}
\tilde{\psi}(Y_1,Y_1',Y_3)=M,\quad
\tilde{\psi}(Y_1,Y_1',Y_4)=M, \\
\tilde{\psi}(Y_2,Y_2',Y_3)=M,\quad
\tilde{\psi}(Y_2,Y_2',Y_4)=M,
\end{align}
where $Y_1'$ and $Y_2'$ are the information replaced by Eve at the edges $e(1)$ and $e(2)$,
and $Y_3$ and $Y_4$ are the information at the edges $e(3)$ and $e(4)$.
This kind of secrecy is called imperfect security \cite{CN11,CN11b}.
\end{description}

Under the above formulation, 
we compare the case of linear codes with the case of non-linear codes.
Since the message is an element of $\ZZ_d$,
the linearity in this problem can be regarded as scalar linearity
when $d$ is a prime number.
We have the following three theorems.

\begin{theorem}\Label{TT5}
There is no secure linear code even for passive attack when $d$ is a prime $p$.
\end{theorem}
\begin{proof}
Due to the linearity of our code, 
we can choose 
a $2 \times 2$ matrix $A$ on $\FF_p$ such that
the relations $ Y_3= A_{1,1}Y_1 +A_{1,2}Y_2$
and
$ Y_4= A_{2,1}Y_1 +A_{2,2}Y_2$
holds when there is no attack.
When the rank of $A$ is 1, 
the information of $Y_3$ is the same as that of $Y_4$.
When Eve eavesdrops $Y_3$, Eve obtains all information that the receiver gets.
Then, Eve recovers the message if the receiver recovers the message.

Next, we assume that the rank of $A$ is 2.
The vector $(A_{1,1},A_{1,2})$ is linearly independent of 
at least one of $(1,0)$ and $(0,1)$.
Assume that it is linearly independent of $(1,0)$, for simplicity.
When Eve eavesdrops $e(3)$ and $e(1)$,
Eve can recover the information $Y_4$ from $Y_3$ and $Y_1$.
Hence, Eve obtains all information that the receiver gets.
\end{proof}

\begin{theorem}\Label{TT6}
When $d=2$, 
there exists a imperfectly secure linear code for passive attack, but
there exists no imperfectly secure linear code for active attack.
\end{theorem}

Since we are allowed to use arbitrary matrices $\bm{K} $ and $\bm{H}$ in Theorem \ref{T1} of the single shot setting,
the security of active attack in the linear codes
can be reduced to the security of passive attack in the linear codes.
Notice that Theorem \ref{T1} can be extended to the case when $\FF_q$ is replaced by $\ZZ_d$.
Theorem \ref{TT6} shows the existence of the non-linear code
whose security for active attack cannot be reduced to 
the security for passive attack.
Hence, such a non-linear code can be regarded as 
a counterexample of Theorem \ref{T1} without the linearity condition.

\begin{theorem}\Label{TT7}
When $d\ge 3$, 
there exists a imperfectly secure linear code even for active attack.
\end{theorem}

\begin{table}[htpb]
  \caption{Summary for one hop relay network (Fig. \ref{F1}) with single shot setting}
\Label{non-linear}
\begin{center}
  \begin{tabular}{|l|c|c|} 
\hline
Code & passive attack & active attack\\
\hline
\hline
linear code when $d$ is prime & insecure & insecure \\
\hline
linear code over $\FF_q$ with &
\multirow{2}{*}{imperfectly secure}& 
\multirow{2}{*}{imperfectly secure} \\
sufficiently large prime power && \\
\hline
non-linear code with $d=2$ & imperfectly secure & insecure \\
\hline
non-linear code with $d\ge 3$ & imperfectly secure & imperfectly secure  \\
\hline
  \end{tabular}
\end{center}
\end{table}

As a summary, we have Table \ref{non-linear}, which shows the  importance of non-linear codes in the prime case.
In fact, this analysis depends on the property of single shot case, i.e., the case with $n=1$.
The next paper \cite{CaiH} discusses the analysis with $n>1$, which contains the vector linearity.

\subsection{Analysis with $d=2$}
\subsubsection{Counterexample of Theorem \ref{T1} without the linearity condition}
First, to show Theorem \ref{TT6},
we give a special example of our code, in which,
the intermediate node performs a non-linear operation as
\begin{align}
Y_3&:= Y_1(Y_2+Y_1)=Y_1(Y_2+1), \Label{Eq1}\\
Y_4&:=(Y_1+1)(Y_2+Y_1)
=(Y_1+1)Y_2. \Label{Eq2}
\end{align}

To send the binary information $M \in \FF_2$,
we prepare the binary uniform scramble random variable $L\in \FF_2$.
We consider the following code.
The encoder $\phi$ is given as 
\begin{align}
Y_1:=L, \quad
Y_2:=M+L. \Label{Eq3}
\end{align}
The decoder $\psi$ is given as $\psi(Y_3,Y_4):=Y_3+Y_4$.
Since $Y_3$ and $Y_4$ are given as follows under this code;
\begin{align}
Y_3= LM, \quad
Y_4=LM+M,
\end{align}
the decoder can recover $M$ whatever the value of $L$.

Now, we consider the leaked information for the passive attack.
As shown in Appendix \ref{B1},
the mutual information and the $l_1$ norm security measure
of these cases are calculated as
\begin{align}
 & I(M;Y_1,Y_3)= I(M;Y_1,Y_4) \nonumber \\
  =& I(M;Y_2,Y_3)=  I(M;Y_2,Y_4)=\frac{1}{2},\Label{E9}\\
 & d_1(M|Y_1,Y_3)= d_1(M|Y_1,Y_4) \nonumber \\
  =& d_1(M|Y_2,Y_3)=  d_1(M|Y_2,Y_4)=\frac{1}{2},\Label{F10}
\end{align}
where the $l_1$ norm security measure $d_1(X|Y)$ is defined as
$d_1(X|Y):=\sum_y \sum_{x}
|\frac{1}{|{\cal X}|}P_Y(y) -P_{XY}(xy)|$
by using the cardinality $|{\cal X}|$ of the set of outcomes of the variable $X$.
In this subsection, we choose $2$ as the base of the logarithm.

\begin{description}
\item[(i)]
When $E_A=E_E=\{e(1),e(3)\}$,
Eve replaces $Y_1$ by $1$. Then, $I(M;Y_1,Y_3)=1 $ because $Y_3+Y_1+1=M$.

\item[(ii)]
When $E_A=E_E=\{e(1),e(4)\}$,
Eve replaces $Y_1$ by $0$. Then, $I(M;Y_1,Y_4)=1 $ because $Y_4+Y_1=M$.

\item[(iii)]
When $E_A=E_E=\{e(2),e(3)\}$ or $\{e(2),e(4)\}$,
Eve has no advantageous active attack.
\end{description}

When Eve is allowed to use the above passive attack, \eqref{E9} shows that the code is secure in the sense of (B1).
Therefore, we obtain the first part of Theorem \ref{TT6}.
Since this code is insecure under the above active attack,
this example is a counterexample of Theorem \ref{T1} without linearity.


\begin{remark}
As another encoder, we can consider 
\begin{align}
Y_1:=M+L, \quad
Y_2:=L.
\end{align}
Replacing $M+L$ by $L$, 
the analysis can be reduced to the presented analysis.
When the message $M$ is not leaked to $e(1)$ or $e(2)$ 
and $M$ can be recovered,
the code is essentially the same as our code as follows.

Assume that the information $Y_2$ on $e(2)$ is independent of $M$.
Then, we denote it by $L$.
In order that $Y_1$ is independent and $M$ can be recovered from $Y_1$ and $L$,
$Y_1$ needs to be $M+L$ or $M+L+1$.
\end{remark}

In this model, Eve can completely contaminate the message $M$.
When Eve takes choice (i), and replaces $Y_3$ by $Y_3+1$,
Bob's decoded message is $M+1$.
Under choice (ii), Eve can totally contaminate the message $M$ in a similar way.

\subsubsection{Uniqueness of network code given in \eqref{Eq1} and \eqref{Eq2}}\label{S2-5}
The previous subsubsection provided an example where Eve's active attack improves her performance.
To show the second part of Theorem \ref{TT6}, 
we need to show the following lemma.


\begin{lemma}\Label{T6}
Assume that a code $(\phi,\varphi,\psi)$ satisfies the following conditions.
Let $Y_1$ and $Y_2$ be the random variable generated by the encoder 
$\phi$ when $M$ is subject to the uniform distribution.
We assume that the random variables $(Y_3,Y_4):=\varphi (Y_1,Y_2)$
satisfies the following conditions.
\begin{description}
\item[(C1)] The relation $\psi (Y_3,Y_4)=M$ holds.
\item[(C2)] There is no deterministic function $\tilde{\psi}$
from $\FF_2^2$ to $\FF_2$ satisfying one of the following conditions.
\begin{align}
\tilde{\psi}(Y_1,Y_3)=M,\quad
\tilde{\psi}(Y_1,Y_4)=M, \\
\tilde{\psi}(Y_2,Y_3)=M,\quad
\tilde{\psi}(Y_2,Y_4)=M.
\end{align}
\end{description}
Then, there exist functions
$f_1,f_2,f_3,f_4$ on $\FF_2$ such that
$Y_i':=f_i(Y_i)$ is given in \eqref{Eq1}, \eqref{Eq2}, and \eqref{Eq3}
with a scramble random variable $L$ while the variable $L$ might be correlated with $M$.
\end{lemma}

Since 
the number of edges to be attacked
is the same as the transmission rate from Alice to Bob,
no linear code works in this scheme.
Hence, we need to introduce a non-linear coding operation in the intermediate node. Lemma \ref{T6} shows that such a non-linear coding operation is limited to 
\eqref{Eq1} and \eqref{Eq2}.

The combination of Lemma \ref{T6} and the discussion in Section \ref{S4}
implies that there is no code over the one hop relay network (Fig. \ref{F1})
to guarantee the secrecy for an active attack.
However, this theorem assumes a deterministic coding operation on the intermediate node.
If we are allowed to use a randomized operation on the intermediate node 
in a similar way to an encoder,
we can construct a code whose secrecy holds even against Eve's active attack
in this situation as follows.
In the first step, we employ the code given in \eqref{Eq3}.
Using another scramble random variable $L'$,
the intermediate node performs the following coding operation:
\begin{align}
Y_3:=Y_1+Y_2+L'=M_L', \quad
Y_4:=L'.
\end{align}
Then, Eve cannot recover the message $M$ from 
$(Y_1,Y_3)$, $(Y_1,Y_4)$, $(Y_2,Y_3)$, nor $(Y_2,Y_4)$.
This example shows that the deterministic condition for $\varphi$ is crucial in Lemma \ref{T6}.

\begin{proofof}{Lemma \ref{T6}}

\noindent{\bf Step (1):}\quad
To satisfy condition (C1),
we need to recover the message $M$ from $(Y_1,Y_2)$ from a deterministic function $f$.
Functions from $\FF_2^2$ to $\FF_2$ are classified as follows.
\begin{align}
&Y_1,~ Y_1+1,Y_2,~ Y_2+1,0,1\Label{D1}\\
&Y_1+Y_2, ~ Y_1+Y_2+1, \Label{D2}\\
&Y_1Y_2,
(Y_1+1)(Y_2+1),
(Y_1+1)Y_2,
Y_1(Y_2+1),
\Label{C3}\\
&
Y_1Y_2+1,
(Y_1+1)(Y_2+1)+1,
(Y_1+1)Y_2+1,
Y_1(Y_2+1)+1.
\Label{C4}
\end{align}
The cases in \eqref{D1} are non-secure or
does not satisfy (C1).
The cases in \eqref{D2} are reduced to the case $M=Y_1+Y_2$.
The cases in \eqref{C3} and \eqref{C4} 
are reduced to the case $M=Y_1Y_2$.
Hence, we consider only these two cases.

\noindent{\bf Step (2):}\quad
Now, we consider the case $M=Y_1+Y_2$.
When $Y_3$ or $Y_4$ is given as a non-zero linear function of $Y_1$ and $Y_2$,
we denote the random variable as $Y_*$.
Hence, $Y_1$ or $Y_2$ is linearly independent of $Y_*$.
We denote the linearly independent variable as $Y_{**}$.
When Eve eavesdrops $Y_*$ and $Y_{**}$, 
she can recover $Y_1$ and $Y_2$ and so she can also recover $M$.
To satisfy condition (C2), we need to avoid such an attack, 
which requires both $Y_3$ and $Y_4$ to be non-linear functions of $(Y_1,Y_2)$.
They are given as two of the functions given in \eqref{C3} and \eqref{C4}.
Since any function in \eqref{C4} is deterministically given from 
a function given in \eqref{C3},
we consider only functions in \eqref{C3}.
Under this constraint,
if and only if $(Y_3,Y_4)$ are given as 
the pair 
$(Y_1Y_2,(Y_1+1)(Y_2+1))$ or $(Y_1(Y_2+1),(Y_1+1)Y_2)$,
we can recover $M=Y_1+Y_2$ from $Y_3$ and $Y_4$.
The latter case is the same as \eqref{Eq1} and \eqref{Eq2}. 
In the former case, we obtain \eqref{Eq1} and \eqref{Eq2}
by replacing $Y_2$ by $Y_2+1$.

\noindent{\bf Step (3):}\quad
Now, we consider the case where $M=Y_1Y_2$.
For the same reason as with Step (2), 
condition (C2) requires 
both $Y_3$ and $Y_4$ to be non-linear functions of 
$(Y_1,Y_2)$.
Thus, we consider only functions in \eqref{C3}.
For secrecy, i.e., to satisfy (C2), we cannot use $Y_1Y_2$.
Hence, we need to choose two from
$(Y_1+1)(Y_2+1),(Y_1+1)Y_2$, and $Y_1(Y_2+1)$.
However, no two of them can recover $M$.
To observe this fact, we consider cases with $(Y_1+1)Y_2 $ and $Y_1(Y_2+1)$.
In these cases, when $(Y_1,Y_2)=(0,0)$ or $(1,1)$,
both values are zero.
That is, we cannot distinguish $(0,0)$ and $(1,1)$.
Hence, we cannot recover $M$ from $(Y_1+1)Y_2 $ and $Y_1(Y_2+1)$, i.e.,
condition (C1) does not hold.
We can show this fact in other pairs in the same way.
Therefore, there is no operation satisfying the required conditions
when $M=Y_1Y_2$.
\end{proofof}

\subsection{Analysis with $d \ge 3$}\label{S3-5}
\subsubsection{Construction of imperfectly secure code for active attacks}\label{S3-5-1}
To show Theorem \ref{TT7}, we construct a secure network coding against any active attack 
on the one hop relay network given in Fig. \ref{F1}
when the message and the information at the edges
are given as an element of $\ZZ_d$.
Here, we define our code $(\phi,\varphi,\psi)$
in the same way as in Subsection \ref{S2-5}.
That is, the coding operation $\varphi$ on the intermediate node
has no additional scramble random variable. 

It show Theorem \ref{TT7}, it is sufficient to construct a code to satisfy the conditions (B1) and (B2) given in Subsection \ref{S4-1}.
Since it is not so easy to check the conditions (B1) and (B2),
we seek equivalent conditions.
For simplicity, 
we employ a scramble variable $L$ taking values in $\ZZ_d$.
Hence, we assume that the encoder $\phi$ in the source node is given as 
a pair of functions $(\phi^{(1)},\phi^{(2)})$
that maps two random variables $(M,L)$
to the two variables $(Y_1,Y_2)$.
That is, the encoder $\phi$ forms a function from $\ZZ_d^2$ to itself.
Now, we fix the function $\phi$ as follows
\begin{align}
Y_1= \phi^{(1)}(M,L):= M+L,\quad 
Y_2= \phi^{(2)}(M,L):=L.\label{con1}
\end{align}

For a coding operation $\varphi$, we define 
the functions $\varphi^{(3)}$ and $\varphi^{(4)}$ as 
\begin{align}
(\varphi^{(3)}(i,j),\varphi^{(4)}(i,j)):=\varphi(i,j).
\end{align}
Then, we regard 
the functions $\varphi^{(3)}$ and $\varphi^{(4)}$ as matrices as follows,
\begin{align}
\varphi^{(3)}_{i,j}:=\varphi^{(3)}(i,j),\quad
\varphi^{(4)}_{i,j}:=\varphi^{(4)}(i,j).
\end{align}
To discuss condition (B2),
we introduce an anti-Latin square.
A matrix $a_{i,j}$ on $\ZZ_d$ is called an {\it anti-Latin square}
when each row and each column have duplicate elements as
\begin{align}
& \left(
\begin{array}{ccc}
1 & 0 & 0 \\
0 & 0 & 2 \\
1 & 2 & 2
\end{array}
\right),\quad
\left(
\begin{array}{ccc}
1 & 0 & 1 \\
1 & 2 & 1 \\
0 & 2 & 0
\end{array}
\right),\Label{Ex1}\\
& \left(
\begin{array}{cccc}
1 & 0 & 3 & 3 \\
0 & 0 & 2 & 3 \\
1 & 1 & 3 & 2 \\
0 & 2& 2 & 1
\end{array}
\right),\quad
\left(
\begin{array}{cccc}
2 & 2 & 1 & 0\\
0 & 3 & 3 & 1 \\
0 & 3 & 3 & 0 \\
1 & 1 & 2 & 2 
\end{array}
\right), \Label{Ex2}
\end{align}
which is the opposite requirement to a Latin square.
Therefore, we have the following lemma.

\begin{lemma}\Label{LOF}
When the encoder $\phi$ satisfies condition \eqref{con1},
conditions (B1) and (B2) are rewritten as
\begin{description}
\item[(B1')] 
For each $a \in \ZZ_d$ and
$m\neq m'\in \ZZ_d$,
the relation $
\Xi_{a,m}(\varphi^{(3)},\varphi^{(4)})
\cap
\Xi_{a,m'}(\varphi^{(3)},\varphi^{(4)})
=\emptyset$ holds,
where $\Xi_{a,m}(\varphi^{(3)},\varphi^{(4)}):=
\varphi^{(4)} (\{(i,i+m)| \varphi^{(3)}_{i,i+m}=a\})$.

\item[(B2')] 
The matrices $\varphi^{(3)}$ and $\varphi^{(4)}$ are anti-Latin squares.
\end{description}
\end{lemma}

\begin{proof}
We have the equivalence between conditions (B1) and (B1')
because (B1') means that the pair of $\varphi^{(3)}(i,j)$ and 
$\varphi^{(4)}(i,j)$ uniquely identifies the difference $m=j-i$.

Next, we show the equivalence between (B2) and (B2').
Assume that Eve eavesdrops and contaminates $Y_1$ and eavesdrops $Y_3$.
Choosing the replaced information $Y_1'$, 
Eve can choose a row of the matrix $\varphi^{(3)}$.
To prevent Eve from recovering $M$ perfectly,
all the rows of the matrix $\varphi^{(3)}$ need to have duplicate elements.
Hence, to satisfy condition (B2),
both matrices $\varphi^{(3)}$ and $\varphi^{(4)}$ need to satisfy this duplication requirement  for all rows and columns.
\end{proof}

Due to Lemma \ref{LOF}, to show Theorem \ref{TT7},
it is sufficient to construct a pair of anti-Latin squares to satisfy conditions (B1') and (B2').
While it is trivial to find anti-Latin squares, they need to satisfy condition (B1') as well.
Condition (B2') forbids a linear operation on the intermediate node in the finite field case.
A pair of anti-Latin squares is called decodable when it satisfies condition (B1').
That is, a decodable pair of anti-Latin squares gives a code on the one hop relay network given in Fig. \ref{F1} satisfying conditions (B1) and (B2).
Lemma \ref{T6} says that there is no decodable pair of $2 \times 2$ anti-Latin squares.
Fortunately, Eq. \eqref{Ex1} (Eq. \eqref{Ex2})
is a decodable pair of  $3 \times 3$ ($4 \times 4$) anti-Latin squares. 

However, we can systematically construct decodable pairs of anti-Latin squares.
The following are pairs of  anti-Latin squares for $d=3,5,7$:
\begin{align}
&\left(
\begin{array}{ccc}
0 & 1 & 0 \\
1 & 1 & 2 \\
0 & 2 & 2
\end{array}
\right),\quad
\left(
\begin{array}{ccc}
0 & 2 & 2 \\
0 & 1 & 0 \\
1 & 1 & 2 
\end{array}
\right),\Label{Ex3}
\\
&\left(
\begin{array}{ccccc}
0 & 1 & 2 & 0 &0 \\
1 & 1 & 2 & 3 & 1\\
2 & 2 & 2 & 3 & 4 \\
0 & 3 & 3 & 3 & 4 \\
0 & 1 & 4 & 4 & 4 
\end{array}
\right),\quad
\left(
\begin{array}{ccccc}
0 & 3 & 3 & 3 & 4 \\
0 & 1 & 4 & 4 & 4 \\
0 & 1 & 2 & 0 &0 \\
1 & 1 & 2 & 3 & 1\\
2 & 2 & 2 & 3 & 4 
\end{array}
\right),\Label{Ex5}
\\
&\left(
\begin{array}{ccccccc}
0 & 1 & 2 & 3 & 0 & 0 &0 \\
1 & 1 & 2 & 3 & 4 & 1 &1 \\
2 & 2 & 2 & 3 & 4 & 5 &2 \\
3 & 3 & 3 & 3 & 4 & 5 & 6 \\
0 & 4 & 4 & 4 & 4 & 5 & 6 \\
0 & 1 & 5 & 5 & 5 & 5 & 6 \\
0 & 1 & 2 & 6 & 6 & 6 & 6 
\end{array}
\right),\nonumber \\
&\left(
\begin{array}{ccccccc}
0 & 4 & 4 & 4 & 4 & 5 & 6 \\
0 & 1 & 5 & 5 & 5 & 5 & 6 \\
0 & 1 & 2 & 6 & 6 & 6 & 6 \\
0 & 1 & 2 & 3 & 0 & 0 &0 \\
1 & 1 & 2 & 3 & 4 & 1 &1\\
2 & 2 & 2 & 3 & 4 & 5 &2 \\
3 & 3 & 3 & 3 & 4 & 5 & 6 
\end{array}
\right).\Label{Ex7}
\end{align}
These constructions are generalized to the case with a general odd number
$d=2\ell+1$ as follows.
The functions 
$\varphi^{(3)}$ 
and
$\varphi^{(4)}$ 
are defined as
\begin{align}
(\varphi^{(3)})^{-1}(k)
&:=
\left\{
\begin{array}{l}
(k,k-\ell),(k,k-\ell+1),\ldots, \\
(k,k-1),(k,k),(k-1,k),\ldots,\\
(k-\ell+1,k),
(k-\ell,k)
\end{array}
\right\} \\
(\varphi^{(4)})^{-1}(k)
&:=
\left\{
\begin{array}{l}
(k+\ell,k-\ell),(k+\ell,k-\ell+1),\\
\ldots, (k+\ell,k-1),\\
(k+\ell,k),(k+\ell-1,k),\\
\ldots,
(k+1,k),
(k,k)
\end{array}
\right\} .
\end{align}
Then, we have
\begin{align}
\varphi^{(4)}(k,k-\ell)&=k-\ell,\\
\varphi^{(4)}(k,k-\ell+1)&=k-\ell+1,\\
\vdots \\
\varphi^{(4)}(k,k-1)&= k-1\\
\varphi^{(4)}(k,k)&=k\\
\varphi^{(4)}(k-1,k)&=k+\ell\\
\vdots\\
\varphi^{(4)}(k-\ell+1,k)&=k+2\\
\varphi^{(4)}(k-\ell,k)&=k+1,
\end{align}
which satisfy condition (B1').
Hence, the functions $\varphi^{(3)}$ and $\varphi^{(4)}$ give a pair of  anti-Latin squares. 

Next, we consider the even case with $d \ge 4$.
The following are pairs of  anti-Latin squares for $d=4,6,8$:
\begin{align}
&\left(
\begin{array}{cccc}
0 & 1 & 3 & 3 \\
0 & 1 & 2 & 0 \\
1 & 1 & 2 & 3 \\
0 & 2& 2 & 3
\end{array}
\right),\quad
\left(
\begin{array}{cccc}
0 & 0 & 1 & 0 \\
1 & 1 & 1 & 2 \\
3 & 2 & 2 & 2 \\
3 & 0& 3 & 3
\end{array}
\right), \Label{Ex4}
\\
&\left(
\begin{array}{cccccc}
0 & 1 & 2 & 5 & 5 &5\\
0 & 1 & 2 & 3 & 0&0\\
1 & 1 & 2 & 3 &  4&1\\
2 & 2& 2 & 3 & 4 & 5\\
0 & 3 & 3 &3 & 4 & 5\\
0 & 1  & 4 & 4 & 4 & 5 
\end{array}
\right),\quad
\left(
\begin{array}{cccccc}
1 & 1 & 1 & 2 & 3 & 1\\
2 & 2 & 2 & 2 & 3& 4\\
5 & 3& 3 & 3 & 3 & 4\\
5 & 0 & 4 &4 & 4 & 4\\
5 & 0  & 1 & 5 & 5 & 5 \\
0 & 0 & 1 & 2 & 0 & 0
\end{array}
\right), \Label{Ex6}
\\
&\left(
\begin{array}{cccccccc}
0 & 1 & 2 & 3 & 7 &7 & 7 &7\\
0 & 1 & 2 & 3 & 4&0 & 0& 0\\
1 & 1 & 2 & 3 &  4&5 & 1 & 1\\
2 & 2& 2 & 3 & 4 & 5 & 6 & 2\\
3 & 3 & 3 &3 & 4 & 5 & 6 & 7\\
0 & 4  & 4 & 4 & 4 & 5 & 6&  7 \\
0 & 1  & 5 & 5 & 5 & 5 & 6&  7 \\
0 & 1  & 2 & 6 & 6 & 6 & 6&  7 
\end{array}
\right),\nonumber \\ 
&\left(
\begin{array}{cccccccc}
2 & 2 & 2 & 2 &  3&4 & 5 & 2\\
3 & 3& 3 & 3 & 3 & 4 & 5 & 6\\
7 & 4 & 4 &4 & 4 & 4 & 5 & 6\\
7 & 0  & 5 & 5 & 5 & 5 & 5 & 6 \\
7 & 0  & 1 & 6 & 6 & 6 & 6&  6 \\
7 & 0  & 1 & 2 & 7 & 7 & 7&  7 \\
0 & 0 & 1 & 2 & 3 &0 & 0 &0\\
1 & 1 & 1 & 2 & 3& 4 & 1& 1\\
\end{array}
\right). \Label{Ex8}
\end{align}
These constructions are generalized to the case with a general even number
$d=2\ell$ with $\ell \ge 2$ as follows.
The functions 
$\varphi^{(3)}$ 
and
$\varphi^{(4)}$ 
are defined as
\begin{align}
(\varphi^{(3)})^{-1}(k)
&:=
\left\{
\begin{array}{l}
(k+1,k-\ell+1),\\
(k+1,k-\ell+2),\ldots, \\
(k+1,k-1),(k+1,k),\\
(k,k),(k-1,k),
\ldots,\\
(k-\ell+2,k),
(k-\ell+1,k)
\end{array}
\right\} \\
(\varphi^{(4)})^{-1}(k)
&:=
\left\{
\begin{array}{l}
(k-\ell+1,k-\ell+2),\\
(k-\ell+2,k-\ell+2),
\ldots,\\ 
(k,k-\ell+2),\\
(k+1,k-\ell+2),\\
(k+1,k-\ell+1),
\ldots,\\
(k+1,k-2\ell+4),\\
(k+1,k-2\ell+3)
\end{array}
\right\} .
\end{align}
Then, we have
\begin{align}
\varphi^{(4)}(k+1,k-\ell+1)&=k+\ell,\\
\varphi^{(4)}(k+1,k-\ell+2)&=k+\ell+1,\\
\vdots \\
\varphi^{(4)}(k+1,k-1)&= k+2\ell-2\\
\varphi^{(4)}(k+1,k)&=k+\ell-1\\
\varphi^{(4)}(k,k)&=k+\ell-2\\
\varphi^{(4)}(k-1,k)&=k+\ell-3\\
\vdots\\
\varphi^{(4)}(k-\ell+2,k)&=k \\
\varphi^{(4)}(k-\ell+1,k)&=k-1,
\end{align}
which satisfy condition (B1').
Hence, the functions $\varphi^{(3)}$ and $\varphi^{(4)}$ give a pair of  anti-Latin squares. 

In summary, since these examples work with $d\ge 3$, 
we have proven Theorem \ref{TT7}, i.e., 
there exists a secure code over the active attacks on the one hop relay network (Fig. \ref{F1})
when $d\ge 3$.

Furthermore, when $\varphi$ is given by these pairs of anti-Latin squares, 
Bob can decode $L$ as well as $M$
while the code given by \eqref{Ex1} nor \eqref{Ex2} cannot.
That is, these systematic constructions work well whenever the encoder
$\phi=(\phi^{(1)},\phi^{(2)})$ is a one-to-one function 
on $\ZZ_d^2$
to satisfy the condition
$
\{i|\exists j, ~\phi^{(1)}{(i,j)}=k\}=
\{i|\exists j, ~\phi^{(2)}{(i,j)}=k\}=
\ZZ_d$ for any $k$.


\subsubsection{Leaked information of our code for passive attacks}
Next, we discuss the leaked information under the above code under passive attacks
in a way similar to \eqref{E9} and \eqref{F10}.

Since for $i=1,2$ and $j=3,4$, the pair $M ,Y_i$ decides $Y_1,Y_2$,
we have $H(Y_j|M Y_i)=0$.
Since $Y_i$ is independent of $M$,
we have
\begin{align}
&I(M;Y_i Y_j)
= H(M)-H(M|Y_iY_j) \nonumber \\
=& H(M|Y_i)-H(M|Y_iY_j)
=I(M; Y_j|Y_i)\nonumber \\ 
=& H(Y_j|Y_i)-H(Y_j|M Y_i)
= H(Y_j|Y_i).
\end{align}
When $d$ is odd,
we have
\begin{align}
&H(Y_j|Y_i)= H(Y_j|Y_i=y_i) \nonumber \\
=& 
\frac{d+1}{2}\cdot
\frac{1}{d}\log \frac{2d}{d+1}
+
\frac{d-1}{2}\cdot
\frac{1}{d}\log d
\end{align}
for any $y_i$ with $i=1,2$ and $j=3,4$.

When $d\ge 4$ and $d$ is even,
we have
\begin{align}
&H(Y_3|Y_2)=H(Y_3|Y_2=y_2)
= H(Y_4|Y_1)= H(Y_4|Y_1=y_1)\nonumber \\
=& 
\frac{d+2}{2} \cdot
\frac{1}{d}\log \frac{2d}{d+2}
+
\frac{d-2}{2}\cdot
\frac{1}{d}\log d \\
&H(Y_3|Y_1)=H(Y_3|Y_1=y_1)= H(Y_4|Y_2)
= H(Y_4|Y_2=y_2)\nonumber \\
=& 
\frac{1}{2}\log 2
+
\frac{1}{2}\log d 
\end{align}
for any $y_1,y_2$.
In summary, when $d$ is large, we have
\begin{align}
I(M;Y_i Y_j) = \frac{1}{2}\log d +\frac{1}{2}\log 2
+ O(\frac{\log d}{d}).\Label{F29-6}
\end{align}

\subsubsection{Lower bound of leaked information for passive attacks}\Label{S3-6}
Next, to show the optimality of the code defined in Subsubsection \ref{S3-5-1}, 
we show that the amount in \eqref{F29-6} is close to the minimum leaked information
under a certain condition when $d$ is large.
To derive a lower bound, we consider the following conditions for our code.

\begin{description}

\item[(D1)]
The coding operation on the intermediate node is deterministic.

\item[(D2)]
Alice can use a scramble random variable $L$.
\end{description}
Since our encoder is given as a stochastic map $\phi$
from $\ZZ_2$ to $\ZZ_2^2$ in Subsection \ref{S4-1},
condition (D2) is a more restrictive condition for our encoder.
Then, we have the following theorem.

\begin{lemma}\label{TD1}
Any network code satisfies 
the inequality
\begin{align}
I(M; Y_i Y_3)+I(M; Y_i Y_4) \ge 2 H(M) - \log d.
\end{align}
for $i=1,2$.
\end{lemma}

Lemma \ref{TD1} shows that
\begin{align}
\max_{i,j}I(M; Y_i Y_j) \ge H(M)- \frac{1}{2}\log d,\Label{NHT}
\end{align}
where the maximum is chosen from $i=1,2$ and $j=3,4$.
That is, to realize $\max_{i,j}I(M; Y_i Y_j)=0$, 
the message $M$ needs to satisfy 
\begin{align}
H(M)\le \frac{1}{2}\log d.
\end{align}
When $M$ is the uniform random variable, \eqref{NHT} can be rewritten as 
\begin{align}
\max_{i,j}I(M; Y_i Y_j) \ge \frac{1}{2}\log d.
\end{align}
This lower bound is almost equal to the RHS of \eqref{F29-6} when $d$ is large.

\begin{proof}

Since $M$ is decoded by $Y_3 Y_4$,
\begin{align*}
& 
H(M| Y_i Y_3)
\le H(Y_3Y_4 | Y_i Y_3)
= H(Y_4 | Y_i Y_3)
 \le H(Y_4| Y_i ).
\end{align*}
Similarly, we have $H(M| Y_i Y_4) \le H(Y_3| Y_i Y_4)$ by replacing $Y_3$ and $Y_4$.
Let $i'$ be the integer $1$ or $2$ that is different from $i$.
Combining them, we have
\begin{align*}
& H(M| Y_i Y_3)+H(M| Y_i Y_4)
 \le H(Y_4| Y_i )+H(Y_3| Y_i Y_4) \\
 =& H(Y_3 Y_4| Y_i ) 
 \stackrel{(a)}{\le}
 H(Y_i Y_{i'}| Y_i ) \\
=& H(Y_{i'}| Y_i ) \le \log d,
\end{align*}
where $(a)$ follows from the fact that 
$Y_3 Y_4$ is decided by $Y_1 Y_2$.
Thus, we obtain
\begin{align*}
&I(M; Y_i Y_3)+I(M; Y_i Y_4) \\
=& 2 H(M) - (H(M| Y_i Y_3)+H(M| Y_i Y_4)) \\
\ge & 2 H(M) - \log d.
\end{align*}
\end{proof}

\section{Conclusion}\Label{SCon}
We have discussed how sequential error injection affects the information leaked to Eve.
As the result, we have shown that 
there is no improvement when 
the network is composed of linear operations.
However, when the network contains non-linear operations,
we have found a counterexample that improves the information obtained by Eve.
Moreover, as Theorem \ref{T3},
we have shown the achievability of the asymptotic rate 
$m_{0}-m_{1}-m_{2}$ for a linear network under the secrecy and robustness conditions
when the transmission rate from Alice to Bob is $m_{0}$,
the rate of noise injected by Eve is $m_{1}$,
and the rate of information leakage to Eve is $m_{2}$.
The converse part of this rate is an interesting open problem.
In addition, 
as Theorem \ref{T3B},
we have discussed the secrecy and the asymptotic transmission rate 
when Eve has a possibility to inject noise into the network.

We have also discussed security over such active attacks on codes with a non-linear operation on the intermediate node in the one hop relay network (Fig. \ref{F1}).
In the binary case, when we impose our code to a certain security condition without an active attack,
as shown in Section \ref{S2-5}, our code is limited to the non-linear code given in Section \ref{S4}.
Unfortunately, the non-linear code given in Section \ref{S4} is insecure under active attacks.
To meet this kind of security condition,
the coding operation on the intermediate node needs to be non-linear.
To characterize this kind of security, we have introduced a new concept an ``anti-Latin square'', which is an opposite concept to a Latin square.
That is, such a secure code can be given as a decodable pair of anti-Latin squares while the concept of ``decodable'' is also introduced in Section \ref{S3-5}.
We have also shown the existence of a decodable pair of 
$d \times d$ anti-Latin squares when $d \ge 3$.
This fact shows that 
there exists a secure code over active attacks in the sense described in Section \ref{S4}
except for the binary case.

Further, we have applied our results to network quantum key distribution.
Then, we have clarified what type of network will enable us to realize secure long distance communication based on short distance quantum key distribution.
However, when we consider only the case given in Fig. \ref{F4},
we can employ a classical (non-quantum) secret sharing protocol \cite{Shamir} instead of network coding 
because all of communications of this case are routing.
In particular, cheater-identifiable secret sharing against rushing cheaters \cite{PW91,RAX,XMT1,XMT2,HK}
enables us to share secure keys without using public channels or prior shared randomness.

In this way, this paper has discussed the application of secure network coding
to a network model whose communications on the edges 
are realized by quantum key distribution.
Replacing the role of quantum key distribution by physical layer security,
we can consider a secure network based on physical layer security.
In particular, 
we can use secure wireless communication \cite{LPS,BC11,SC,Trappe,Zeng,WX,H17}
as a typical form of physical layer security,
which provides us with a secure network based on secure wireless communication.
A crucial weak point of physical layer security is the possibility that the eavesdropper might break the assumption of the model.
Such an attack might be realized in the following cases.
(1) The eavesdropper concentrates his/her resources on one point.
(2) The eavesdropper luckily encounters a situation that the assumption is broken.
When we combine physical layer security and secure network coding in the above way,
to eavesdrop our information, 
the eavesdropper needs to break the model of physical layer security in multiple communication channels.
In case (1), to realize this condition, the eavesdropper has to 
distribute his/her resources, which increases the difficulty of eavesdropping.
For case (2), 
the eavesdropper must be lucky in multiple communication channels,
and this probability is very small.
In this way, this kind of combination is particularly useful.

\section*{Acknowledgments}
MH and NC are very grateful to Dr. Wangmei Guo
and Mr. Seunghoan Song
for helpful discussions and comments.
The studies reported here were supported in part by 
the JSPS Grant-in-Aid for Scientific Research 
(C) No. 16K00014, (B) No. 16KT0017, (A) No.17H01280, 
(C) No. 17K05591, 
the Okawa Research Grant,
and Kayamori Foundation of Informational Science Advancement.

\appendices

\section{Proof of Proposition \ref{T2}}\Label{Ap1}

To show Proposition \ref{T2}, 
we regard any element of the finite field $\FF_q$ as
an element of a $t$-dimensional algebraic extension $\FF_{q'}$ 
of the finite field $\FF_q$, where $q'=q^t$.
The matrices $K_B,H_B,K_E,H_E$ on $\FF_q$ can be regarded as matrices on $\FF_{q'}$.
By choosing $l_n:=t n$, the matrices $X^{l_n}$, $Y_B^{l_n}$, $Y_E^{l_n}$, $Z^{l_n}$ on $\FF_q$ 
are converted to matrices ${X'}^{n}$, ${Y_B'}^{n}$, ${Y_E'}^{n}$, ${Z'}^{n}$ on $\FF_{q'}$,
which also satisfy \eqref{F1n} and \eqref{F2n} by regarding 
the same matrices $K_B,H_B,K_E,H_E$ on $\FF_q$ as matrices on $\FF_{q'}$.
Then, the following proposition is known.

\begin{proposition}[\cite{JL,JLKHKM,Jaggi2008,Yao2014}]
\Label{T2B}
We assume the following two conditions for $m_2, m_1, m_0$ and a sequence of prime power $q_{n}'$.
The inequality $m_{2}+m_{1}< m_{0}$ holds.
The size $q_{n}'$ of the finite field increases such that 
$\frac{q_{n}'}{{n}^{m_0+1}} \to \infty$. 
Then, there exists a sequence of codes $\Phi_{n}$ 
of block-length $n$ on finite field $\FF_{q_{n}'}$
whose message set is $\FF_{q_{n}'}^{k_{n}'}$ such that
\begin{align}
&\lim_{n \to \infty} \frac{k_{n}'}{n} = m_{0}-m_{1}\Label{eqH4-27C}
\\
&\lim_{n \to \infty} \max_{ \bm{K},\bm{H}} \max_{\alpha^n} 
P_e[\Phi_{n}, \bm{K},\bm{H},\alpha^n]=0 \Label{eqH4-27},
\end{align}
where the maximum is taken in the same way as with Proposition \ref{T2}.
\end{proposition}

The optimality of the rate $m_0-m_{1}$ was also shown 
under the condition \eqref{eqH4-27} in \cite[Sections VI \& VII]{Yao2014}. 
By choosing $t_{n}=\lceil \frac{ (m_0+1)\log n}{\log q}\rceil$
and $l_{n}:= t_{n} {n}$,
Proposition \ref{T2B} implies Proposition \ref{T2}.
Hence, we need to explain how to show Proposition \ref{T2B}. 

Combining the results in \cite{JL,JLKHKM,Jaggi2008,Yao2014},
we can construct a sequence of codes to satisfy \eqref{eqH4-27C} and \eqref{eqH4-27}.
More precisely, the papers \cite[Section IX]{JLKHKM}\cite[Section VIII]{Jaggi2008}
constructed a sequence of codes to satisfy \eqref{eqH4-27C} and \eqref{eqH4-27}
under  the condition $m_{2}+2 m_{1}< m_{0}$.
This is because the condition $m_{2}+2 m_{1}< m_{0}$ is stronger than 
the condition $m_{2}+m_{1}< m_{0}$, which is the assumption of Proposition \ref{T2B}.
To show \eqref{eqH4-27} under the weaker condition $m_{2}+m_{1}< m_{0}$, 
the papers 
by Jaggi, Langberg, Katti, Ho, Katabi, M\'{e}dard, and Effros
\cite[Section VII]{JLKHKM}\cite[Section VI]{Jaggi2008}\cite[Section IV-C]{Yao2014}
constructed a sequence of codes to satisfy \eqref{eqH4-27C} and \eqref{eqH4-27},
when Alice can send Bob secret information whose size is asymptotically negligible 
in comparison with $n$, in the following way.

\begin{proposition}[\cite{JL,JLKHKM,Jaggi2008}]
\Label{T2C}
We assume the following three conditions.
The inequality $m_0>m_1$ holds.
Alice can send Bob secret information whose size is asymptotically negligible 
in comparison with $n$.
The size $q_{n}'$ of the finite field increases such that 
$\frac{q_{n}'}{{n}^{m_0+1}} \to \infty$.
Then, there exists a sequence of codes $\Phi_{n}$ 
of block-length $n$ on finite field $\FF_{q_{n}'}$
whose message set is $\FF_{q_{n}'}^{k_{n}'}$ 
such that the relations \eqref{eqH4-27} and \eqref{eqH4-27C} hold. 
\end{proposition}

Then, under the weaker condition $m_{2}+ m_{1}< m_{0}$,
as the following proposition, the papers \cite[Section III]{JL}\cite[Section V]{Yao2014} provide a protocol for secure transmission of random variables with an asymptotically negligible length $k_{n}$ in comparison with $n$, which is the requirement in Proposition \ref{T2C}.

\begin{proposition}[\protect{\cite[Section III]{JL},\cite[Section V]{Yao2014}}]\Label{T2BL}
We assume the inequality $m_{2}+m_{1}< m_{0}$.
Then, there exists a sequence of codes $\Phi_{n}$ of block-length $n$ 
whose message set is $\FF_{q}^{k_{n}}$ such that
\begin{align}
&\lim_{n \to \infty} k_{n}= \infty \Label{eqH4-27CY}
\\
&\lim_{n \to \infty} \max_{ \bm{K},\bm{H}} \max_{\alpha^n} 
P_e[\Phi_{n}, \bm{K},\bm{H},\alpha^n]=0 \Label{eqH4-27Y},\\
&\lim_{n \to \infty} \max_{ \bm{K},\bm{H}} 
I(M;Y_E^{n})[\Phi_{n},\bm{K},\bm{H},0]=0, \Label{H3-182Y}
\end{align}
where the maximum is taken in the same way as Proposition \ref{T2B}.
\end{proposition}

Therefore, attaching the protocol of Proposition \ref{T2BL} to the codes given in 
Proposition \ref{T2C},
we obtain \eqref{eqH4-27} under the weaker condition $m_{2}+m_{1}< m_{0}$.
However, their proof of Proposition \ref{T2C}
is very hard to read because it omits the detail derivation. 
In the following, we give an alternative proof of Proposition \ref{T2C}.

Before our proof of Proposition \ref{T2C},
we prepare two lemmas.
The first lemma can be easily shown by the discussion of linear algebra.
In the following discussion, we simplify $q_{n}'$ to $q'$.

\begin{lemma}\Label{L4-27}
For integers $a_0 \le a_1+a_2,a_1',a_2'$, 
we fix an $a_1$-dimensional subspace $W_1 \subset \FF_{q'}^{a_1'}$ and 
an $a_2$-dimensional subspace $W_2 \subset \FF_{q'}^{a_2'}$.
We assume the following two conditions.
\begin{description}
\item[(E1)]
An $a_0 \times a_1'$ matrix $A_1$
and an $a_0 \times a_2'$ matrix $A_2$
satisfy 
\begin{align}
\Ker A_1|_{W_1} &=\{0\}, \\
\im A_1 \cap \im A_2 &= \{ 0 \},
\end{align}
where $\im (f)$ denotes the image of the function $f$.
\item[(E2)]
We consider a subspace $W_3 \subset W_1 \oplus W_2$.
For vectors
$x_1, \ldots, x_b \in W_1 $
and
$y_1, \ldots, y_b \in W_2$ with $b \ge a_1+a_2$,
$(x_1,y_1), \ldots, (x_b ,y_b) $ span $W_3$.
\end{description}
Then, we have the following statements.
\begin{description}
\item[(E3)]
There exists an $a_1' \times a_0 $ matrix $A_3$ such that
$ A_3 (A_1 x_i+ A_2 y_i)=x_i$ for $i=1, \ldots, b$, i.e., 
\begin{align}
A_3 
\left[A_1~ A_2 \right]
\left[
\begin{array}{cccc}
x_1 & x_2 & \cdots & x_b \\
y_1 & y_2 & \cdots & y_b
\end{array}
\right]
=
\left[
\begin{array}{cccc}
x_1 & x_2 & \cdots & x_b 
\end{array}
\right].
\end{align}
\item[(E4)]
The above matrix $A_3$ satisfies the relation
\begin{align}
 A_3 (A_1 x+ A_2 y)=x \Label{4-27F}
\end{align}
for any $(x,y) \in W_3$.
\end{description}
\end{lemma}

\begin{proof}
Due to condition (E1), 
we choose a map $A_4$ from $\im A_1$ to $W_1$
such that $A_4 A_1$ is the identify on $W_1$.
Since $W_1$ is included in $\FF_{q'}^{a_1'}$,
$A_4$ can be regarded as a map from $\im A_1$ to $\FF_{q'}^{a_1'}$.
Then, we choose a projection $A_5$ from $\FF_{q'}^{a_0}$
to $\im A_1$ such that $A_5 x=0$ for $x \in \im A_2$.
Therefore, $A_3:= A_5 A_4$ satisfies the condition of (E3).
Further, (E2) guarantees (E4).
\end{proof}


\begin{lemma}[\protect{\cite[Section VII]{JLKHKM}\cite[Claim 5]{Jaggi2008}}]\Label{LJa}
We independently choose $m$ random variables 
$V_1, \ldots, V_{m}$ 
subject to the uniform distribution on $\FF_{q'} $.
We define the $n \times m$ matrix $U_1$
as $U_{1;i,j}:= (V_j)^i$ with $i=1, \ldots, n$ and $j= 1,\ldots, m$.
Then,  
\begin{align}
\Pr \{ x U_1=x' U_1
\} \le \Big(\frac{n}{q'}\Big)^m
\end{align}
for any $x\neq x'\in \FF_{q'}^n$.
\end{lemma}

\begin{proofof}{Proposition \ref{T2C}}\par
\noindent{\bf Step (1): Code construction}\par\noindent
First, we provide our code when we use the channel $n$ times based on the finite field $\FF_{q'}$.
Our message is given as an $(m_{0}-m_{1}) \times n$ matrix $M$, which satisfies condition
\eqref{eqH4-27C} asymptotically. 
Since the rank of $H_B$ is $m_1$,
there exist a 
$m_4\times m_1$ matrix $\hat{H}_B$ and 
$m_1\times n$ matrix $\hat{Z}^{n}$ such that
\begin{align}
Y_B^{n}
=K_B U_0  X^{n}+H_B Z^{n}
=K_B U_0  X^{n}+\hat{H}_B \hat{Z}^{n}.\Label{TT}
\end{align}
Then, we address $\hat{H}_B $ and $\hat{Z}^{n}$
instead of ${H}_B $ and ${Z}^{n}$.

We fix an integer $m:=m_0+1$.
We independently choose $m$ random variables 
$V_1, \ldots, V_{m}$ subject to the uniform distribution on $\FF_{q'} $.
Also, we randomly choose 
the $m_3 \times m_3$ matrix $U_0$ among all $m_3 \times m_3$ invertible matrices. 

Then, we define the $n \times m$ matrix $U_1$
as $U_{1;i,j}:= (V_j)^i$ with $i=1, \ldots, n$ and $j= 1,\ldots, m$.
We also define the $(m_{0}-m_{1})\times m$ matrix $U_2:= M U_1$.
Moreover, we define the $m_3 \times n $ matrix $X^{n}:=
\left[
\begin{array}{c}
M \\
0
\end{array}
\right]$, where $0$ is the $m_1 \times n$ zero matrix.
As secret information with a negligible rate,
Alice sends Bob the information $V_1, \ldots, V_{m},U_2$.
Then, Alice inputs the $m_3 \times n $ matrix $U_0 X^{n}$ 
as the input of $n$ times use of the channel.

Then, Bob receives the $m_4 \times n$ matrix $Y_B^{n}$ given in \eqref{TT}
as well as the secret information $V_1, \ldots, V_{m},U_2$.
Since the ranks of $K_B U_0 X^{n}$ and $\hat{H}_B$ are $m_{0}-m_{1}$ and $m_1$ at most, respectively,
the rank of the matrix $Y_B^{n}$ is $m_0$ at most. 
We denote the rank by $\bar{m}_0$.
We choose $\bar{m}_0$ linearly independent row vectors from the row vectors of $Y_B^{n}$.
We denote the $\bar{m}_0 \times n$ matrix composed of the $\bar{m}_0 $ 
independent row vectors by $\bar{Y}_B^{n} $.
Similarly, we denote the matrices composed of these $\bar{m}_0$ row vectors of 
the matrices $K_B$ and $\hat{H}_B$ by
$\bar{K}_B$ and $\bar{H}_B$, respectively.
Then, using the standard Gaussian elimination,
Bob finds a matrix $ U_3$ to satisfy the equation
\begin{align}
U_3 \bar{Y}_B^{n} U_1= U_2,
\Label{4-27Y}
\end{align}
which is equivalent to
\begin{align}
U_3 
 \left[
\bar{K}_BU_0 P_{m_0-m_1} ~ \bar{H}_B 
\right]
 \left[
\begin{array}{c}
M \\
\hat{Z}^{n}
\end{array}
\right] U_1=M U_1,\Label{EEQ1}
\end{align}
where 
$P_{m_0-m_1}$ is the imbedding
$\left[
\begin{array}{c}
I \\
0\end{array}
\right]
 $
from $\FF_{q'}^{m_0-m_1} $ to 
$\FF_{q'}^{m_3} $.
Notice that Bob can calculate 
$U_1$ from the secret information $V_1, \ldots, V_{m}$.
Finally, Bob recovers the information $\hat{M}:=U_3 \bar{Y}_B^{n}$.
To check the condition \eqref{4-27Y}, Bob needs only 
$\bar{Y}_B^{n}$,
$U_2$, and $ U_1$, which can be computed from 
$V_1, \ldots, V_{m}$.

\noindent{\bf Step (2): Analysis of performance}\par\noindent
There are two conditions if the above protocol is to work well.
\begin{description}
\item[(F1)]
The relations 
$\im (\bar{K}_B U_0 P_{m_0-m_1})\cap \im \bar{H}_B= \{0\}$
and $\Ker \bar{K}_B U_0 P_{m_0-m_1}|_{\im M} =\{0\}$ hold.
\item[(F2)]
The relation 
$
\rank  
 \left[
\begin{array}{c}
M \\
\hat{Z}^{n}
\end{array}
\right]
 U_1
=\rank 
\left[
\begin{array}{c}
M \\
\hat{Z}^{n}
\end{array}
\right]
$ holds,
where $\rank $ denotes the rank of the matrix.


\end{description}

Assume that conditions (F1) and (F2) hold.
We apply Lemma \ref{L4-27} to the case when
$a_0= \bar{m}_0,
a_1=\rank M,
a_2= \rank \hat{Z}^{n},
W_1=\im M,
W_2= \im \hat{Z}^{n},
A_1=\bar{K}_B U_0 P_{m_0-m_1},
A_2=\bar{H}_B,
A_3= U_3$.
Then, conditions (F1) and (F2) guarantee
conditions (E1) and (E2), respectively.
Then, due to condition (E3),
there exists a matrix $ U_3$ that satisfies equation \eqref{4-27Y}, i.e., \eqref{EEQ1}.
Condition (E4) guarantees that
$U_3 \bar{Y}_B^{n}= M$, i.e., 
Bob can decode the message $M$.

Now, we evaluate the probability that condition (F2) holds.
Condition (F2) holds if and only if 
$ z^T  \left[
\begin{array}{c}
M \\
\hat{Z}^{n}
\end{array}
\right]
 U_1 \neq 0$ for any $z \in \FF_{q'}^{m_0}$
 satisfying the condition $
 z^T  \left[
\begin{array}{c}
M \\
\hat{Z}^{n}
\end{array}
\right]
\neq 0 $.
Applying Lemma \ref{LJa} to all of $z (\neq 0)\in \FF_{q'}^{m_0}$,
we find that condition (F2) holds at least with probability 
$1- {q'}^{m_0} (\frac{n}{q'})^m
=1- \frac{{n}^m}{{q'}^{m-m_0}}=1- \frac{{n}^{m_0+1}}{{q'}} \to 1$.

Finally, we evaluate the probability that condition (F1) holds.
As shown later, the following conditions (F1') and (F1'') imply condition (F1).
\begin{description}
\item[(F1')]
The relation $\im (K_B U_0 P_{m_0-m_1})\cap \im \hat{H}_B= \{0\}$ holds.
\item[(F1'')]
The relation $\Ker {K}_B U_0 P_{m_0-m_1}|_{\im M} =\{0\}$ holds.
\end{description}
Hence, we show that conditions (F1') and (F1'') hold with a probability close to $1$.
Condition $\im (K_B U_0 P_{m_0-m_1})\cap \im \hat{H}_B= \{0\}$ holds if and only if
$ \im U_0 P_{m_0-m_1}\cap K_B^{-1} (\im \hat{H}_B)= \{0\}$.
For a fixed $K_B,\hat{H}_B$, 
since $\dim K_B^{-1} (\im \hat{H}_B) \le m_3-m_0+m_1$,
the probability of 
condition $\im (K_B U_0 P_{m_0-m_1})\cap \im \hat{H}_B= \{0\}$ is at least
\begin{align}
& (1-{q'}^{m_3-m_0+m_1-m_3})
(1-{q'}^{m_3-m_0+m_1-m_3+1})\cdots
(1-{q'}^{m_3-m_0+m_1-m_3+ m_0-m_1-1}) \nonumber \\
=&(1-{q'}^{m_1-m_0})
(1-{q'}^{m_1-m_0+1})\cdots
(1-{q'}^{-1})
=1-O(1/q').
\end{align}

The relation $\Ker {K}_B U_0 P_{m_0-m_1}|_{\im M} =\{0\}$ holds if and only if
no basis of $\im U_0 P_{m_0-m_1} M$ belongs to the space $\Ker {K}_B $.
Since $\Ker {K}_B $ is an $m_3-m_0$-dimensional subspace of an $m_3$-dimensional space,
the probability of condition 
$\Ker {K}_B U_0 P_{m_0-m_1}|_{\im M} =\{0\}$ is
\begin{align}
(1-{q'}^{-m_0})
(1-{q'}^{-m_0+1})\cdots
(1-{q'}^{-m_0+m_7-1})
=1-O({q'}^{-m_0+m_7-1}),
\end{align}
where $m_7:= \dim \im P_{m_0-m_1} M(= \dim \im U_0  P_{m_0-m_1} M) 
=\dim \im M= \rank M \le m_0-m_1$.
Therefore, since $q'$ is sufficiently large,
we obtain the desired statement.

Finally, we show that condition (F1') implies condition (F1).
Since the relation
$\im (\bar{K}_B U_0 P_{m_0-m_1})\cap \im \bar{H}_B= \{0\}$
can be shown easily from (F1'), we show only the relation
$\Ker \bar{K}_B U_0 P_{m_0-m_1}|_{\im M} =\{0\}$ from (F1') and (F1'').
The choice of $\bar{m}_0$ guarantees that 
there exists an invertible map $U_4$ from $\im [K_B U~ \hat{H}_B]$ to $\FF_{q'}^{\bar{m}_0}$ such that
$U_4 [K_B U~ \hat{H}_B]=[\bar{K}_B U~ \bar{H}_B]$. 
Thus, 
\begin{align}
\Ker \bar{K}_B U_0 P_{m_0-m_1}|_{\im M} =
\Ker U_4^{-1} \bar{K}_B U_0 P_{m_0-m_1}|_{\im M} =
\Ker {K}_B U_0 P_{m_0-m_1}|_{\im M} =\{0\}.
\end{align}
\end{proofof}

\begin{remark}
Our proof is different from the proof presented in \cite[Section VII]{JLKHKM}\cite[Section VI]{Jaggi2008}\cite[Section IV-C]{Yao2014}.
They suggested that $ (m_0-m_1)m_0+1$ be chosen as $m$ 
because they employ the concept of list decoding.
However, our discussion allows us to choose a much smaller value $m_0+1$
as $m$.
This fact shows that our evaluation is better than their evaluation in this sense.
Note that our evaluation does not use list decoding.
\end{remark}

Proposition \ref{T2B} requires a finite field $\FF_{q'}$ with an infinitely large $q'$.
The paper \cite[Appendix D]{Hayashi-Tsurumaru} discussed the construction of $\FF_{2^t}$ whose
multiplication and inverse multiplication have calculation complexity $O(t \log t)$\footnote{The multiplication of elements $v$ and $z$ of $\FF_{2^t}$ is essentially given in (124) of \cite{Hayashi-Tsurumaru}
by using Fourier transform via a calculation on circulant matrices.
For the inverse multiplication of an element $v$ of $\FF_{2^t}$, 
we calculate $ F^{-1}[-Fv. *Fz]$ instead of 
$ F^{-1}[Fv. *Fz]$ in (124), where $F$ is discrete Fourier transform.}.

\section{Security analysis of passive attack in one hop relay network (Fig. \ref{F1})}\Label{B1}
First, we calculate $I(M;Y_1,Y_3)$ and $d_1(M|Y_1,Y_3)$.
We find that
\begin{align*}
P_{Y_1,Y_3|M}(0,0|0)=
P_{Y_1,Y_3|M}(1,0|0) =\frac{1}{2},\\
P_{Y_1,Y_3|M}(0,0|1)=
P_{Y_1,Y_3|M}(1,1|1)=\frac{1}{2},
\end{align*}
where the remaining conditional probabilities are zero.
Hence, 
\begin{align*}
H(Y_1,Y_3|M)=1, ~
H(Y_1,Y_3)=\frac{1}{2}\log 2 +\frac{1}{2}\log 4 =\frac{3}{2},
\end{align*}
which implies $I(M;Y_1,Y_3)=\frac{1}{2}$.

Since
\begin{align*}
P_{Y_1,Y_3}(0,0)=\frac{1}{2},
P_{Y_1,Y_3}(1,0)=
P_{Y_1,Y_3}(1,1)=\frac{1}{4}, 
\end{align*}
we have
\begin{align*}
P_{M|Y_1,Y_3}(0|0,0)&=
P_{M|Y_1,Y_3}(1|0,0)=\frac{1}{2},\\
P_{M|Y_1,Y_3}(0|1,0)&=
P_{M|Y_1,Y_3}(1|1,1)=1,
\end{align*}
where the remaining conditional probabilities are zero.
Therefore,
\begin{align*}
&d_1(M|Y_1,Y_3) \\
= &
\Big|\frac{1}{2}P(0,0)-\frac{1}{2}P(0|0,0)\Big|
+\Big|\frac{1}{2}P(1,0)-\frac{1}{2}P(0|1,0)\Big|
+\Big|\frac{1}{2}P(1,1)-\frac{1}{2}P(0|1,1)\Big| \\
&+\Big|\frac{1}{2}P(0,0)-\frac{1}{2}P(1|0,0)\Big|
+\Big|\frac{1}{2}P(1,0)-\frac{1}{2}P(1|1,0)\Big|
+\Big|\frac{1}{2}P(1,1)-\frac{1}{2}P(1|1,1)\Big|\\
=&0+\frac{1}{8}+\frac{1}{8}
+0+\frac{1}{8}+\frac{1}{8}
=
\frac{1}{2}.
\end{align*}
Replacing $M$ and $L$ by $M+1$ and $L+1$, respectively,
we can calculate $I(M;Y_1,Y_4)$ and $d_1(M|Y_1,Y_4)$
in the same way.

Next, we consider 
$I(M;Y_2,Y_3)$ and $d_1(M|Y_2,Y_3)$.
We find that
\begin{align*}
P_{Y_2,Y_3|M}(0,0|0)=
P_{Y_2,Y_3|M}(1,0|0) =\frac{1}{2},\\
P_{Y_2,Y_3|M}(0,1|1)=
P_{Y_2,Y_3|M}(1,0|1)=\frac{1}{2},
\end{align*}
where the remaining conditional probabilities are zero.
Hence, replacing $(0,0)$ and $(1,1)$ by $(1,0)$ and $(0,1)$, respectively, in the above
derivation, we can show 
$I(M;Y_2,Y_3)=\frac{1}{2}$
and $d_1(M|Y_2,Y_3)=\frac{1}{2}$.
Finally, replacing $M$ and $L$ by $M+1$ and $L+1$, respectively,
we can calculate $I(M;Y_2,Y_4)$ and $d_1(M|Y_2,Y_4)$ in the same way.

\end{document}